\documentclass[11pt]{article}%
\usepackage{amsmath}
\usepackage{amssymb}
\usepackage{amsfonts}
\usepackage{sw20elba}
\usepackage[longnamesfirst]{natbib}
\usepackage{graphicx}
\usepackage{hyperref}
\usepackage{hyphenat}
\usepackage{afterpage}%
\providecommand{\U}[1]{\protect\rule{.1in}{.1in}}
\newtheorem{theorem}{Theorem}[section]
\newtheorem{assumption}{Assumption}[section]

\newtheorem{corollary}{Corollary}

\newtheorem{lemma}{Lemma}[section]

\newtheorem{proposition}[theorem]{Proposition}

\newenvironment{proof}[1][Proof]{\noindent \textbf{#1.} }{\  \rule{0.5em}{0.5em}}
\newcommand{\GG}[1]{}
\begin{document}

\title{An Asymptotically F-Distributed Chow Test in the Presence of
Heteroscedasticity and Autocorrelation\thanks{We thank Derrick H. Sun for
excellent research assistance. }}
\author{Yixiao Sun\\Department of Economics\\UC San Diego, USA
\and Xuexin Wang\\School of Economics and WISE\\Xiamen University, China}
\maketitle

\begin{abstract}
This study proposes a simple, trustworthy Chow test in the presence of
heteroscedasticity and autocorrelation. The test is based on a series
heteroscedasticity and autocorrelation robust variance estimator with
judiciously crafted basis functions. Like the Chow test in a classical normal
linear regression, the proposed test employs the standard F distribution as
the reference distribution, which is justified under fixed-smoothing
asymptotics. Monte Carlo simulations show that the null rejection probability
of the asymptotic F test is closer to the nominal level than that of the
chi-square test.

\medskip

\medskip


\end{abstract}

\section{Introduction}

For predictive modeling and policy analysis using time series data, it is
important to check whether a structural relationship is stable over time. The
\citet{CHOW1960} test is designed to test whether a break takes place at a
given period in an otherwise stable relationship. The test is widely used in
empirical applications and has been included in standard econometric
textbooks. This paper considers the Chow test in the presence of
heteroscedasticity and autocorrelation. There is ample evidence that the Chow
test can have very large size distortions if heteroscedasticity and
autocorrelation are not accounted for (e.g., \citet{KRAMER1989} and
\citet{GILES1992}). Even if we account for them using heteroscedasticity and
autocorrelation robust (HAR) variance estimators (e.g., \cite{NW1987} and
\cite{Andrews1991}), the test can still over-reject the null hypothesis by a
large margin if chi-square critical values are used\footnote{When the Chow
test is performed on a single coefficient, normal critical values are
typically used on the t statistic. For now, we focus only on the Wald-type
Chow test for more than one coefficients so that chi-square critical values
are used.}. This is a general problem for any HAR inference, as the chi-square
approximation ignores the often substantial finite sample randomness of the
HAR variance estimator. To address this problem, the recent literature has
developed a new type of asymptotics known as fixed-smoothing asymptotics (see,
e.g., \cite{KV2002a, KV2002b, KV2005} for early seminal contributions). It is
now well known that the fixed-smoothing asymptotic approximation is more
accurate than the chi-square approximation. This has been confirmed by ample
simulation evidence and supported by higher-order asymptotic theory\ in
\citet{J2004} and
\citet*{SPJ2008}%
.

In this study, we employ the series HAR variance estimator to implement the
Chow test in a time series regression where the regressors and regression
errors are autocorrelated. This type of HAR variance estimator is the series
counterpart of the kernel HAR variance estimator. The advantage of using the
series HAR variance estimator is that we can design the basis functions so
that the fixed-smoothing asymptotic distribution is the standard F
distribution. This is in contrast to commonly used kernel HAR variance
estimators where the fixed-smoothing asymptotic distributions are nonstandard
and critical values have to be simulated\footnote{In the series case,
fixed-smoothing asymptotics holds the number of basis functions fixed as the
sample size increases. In the kernel case, fixed-smoothing asymptotics holds
the truncation lag parameter fixed at a certain proportion of the sample
size.}.

To establish the asymptotic\ F theory for the Chow test under fixed-smoothing
asymptotics, we have to transform the usual orthonormal bases such as sine and
cosine bases using the Gram--Schmidt orthonormalization. This is because,
unlike the HAR inference in a regression with stationary regressors and
regression errors, using the usual bases as in
\citet{S2013}
does not lead to a standard fixed-smoothing asymptotic distribution, since the
regressors in the regression for the structural break test are identically
zero before or after the break point and are thus not stationary. The
Gram--Schmidt orthonormalization ensures that the transformed bases are
orthonormal with respect to a special inner product that is built into the
problem under consideration. The asymptotic F test is very convenient to use,
as the F critical values are readily available from standard statistical
tables and programming environments.

Monte Carlo simulation experiments show that the F test based on the
transformed Fourier bases is as accurate as the nonstandard test based on the
usual Fourier bases. The F test and nonstandard test have the same
size-adjusted power as the corresponding chi-square tests but much more
accurate size. Given its convenience, competitive power, and higher size
accuracy, we recommend the F test for practical use.

Our F test theory generalizes the classical Chow test in a linear normal
regression where the F distribution is the exact finite sample distribution.
The main departures are that we do not make the normality assumption and that
we allow for heteroscedasticity and autocorrelation of unknown forms. Without
restrictive assumptions such as normality and strict exogeneity, it is in
general not possible to obtain the exact finite sample distribution. Instead,
we employ the fixed-smoothing asymptotics to show that the Wald statistic is
asymptotically F distributed.

This study contributes to the asymptotic F test theory in the HAR literature.
The asymptotic F theory has been developed in a number of papers including
\citet{S2011, SK2012, S2013, HS2017, LLSW2018, LS2018, WS2019, MSW2019}.
However, none of these studies considers the case where the regressors take
the special form of nonstationarity as we consider here.
\citet{Cho_Vogelsang2017} consider fixed-b asymptotics for testing structural
breaks, but they consider only kernel HAR variance estimators. As a result,
the fixed-smoothing asymptotic distributions they obtained are highly nonstandard.

The rest of this paper is organized as follows. Section \ref{basic_setting}
presents the basic setting and introduces the test statistics. Section
\ref{FSA} establishes the fixed-smoothing asymptotics of the F and t
statistics. Section \ref{Sec: F and t} develops asymptotically valid F and t
tests. Section \ref{Sec: Extend} extends the basic regression model to include
other covariates whose coefficients are known to be stable over time. Section
\ref{Sec: Simulation} reports the simulation evidence. The last section
concludes. Proofs are given in the appendix.

\section{Basic Setting and Test Statistics\label{basic_setting}}

Given the time series observations $\left\{  X_{t}\in\mathbb{R}^{m},Y_{t}%
\in\mathbb{R}\right\}  _{t=1}^{T},$ we consider the model%
\[
Y_{t}=X_{t}\cdot1\left\{  t\leq\left[  \lambda T\right]  \right\}  \cdot
\beta_{1}+X_{t}\cdot1\left\{  t\geq\left[  \lambda T\right]  +1\right\}
\cdot\beta_{2}+u_{t},
\]
for $t=1,2,\ldots,T$ where the unobserved $u_{t}$ satisfies $E\left(
X_{t}u_{t}\right)  =0.$ In the above, $\lambda$ is a known parameter in
$(0,1)$ so that $\left[  \lambda T\right]  $ is the period where the
structural break may take place. The effects of $X_{t}$ on $Y_{t}$ before and
after the break are $\beta_{1}\in\mathbb{R}^{m}$ and $\beta_{2}\in
\mathbb{R}^{m},$ respectively. We allow $X_{t}u_{t}$ to exhibit
autocorrelation of unknown forms. In particular, we allow $u_{t}$ to be
heteroskedastic so that $E(u_{t}^{2}|X_{t})$ is a nontrivial function of
$X_{t}.$

We are interested in testing the null of $H_{0}$: $\mathcal{R}\beta
_{1}=\mathcal{R}\beta_{2}$ against the alternative $H_{1}:$ $\mathcal{R}%
\beta_{1}\neq\mathcal{R}\beta_{2}$ for some $p\times m$ matrix $\mathcal{R}.$
When $\mathcal{R}$ is the $m\times m$ identity matrix, we aim at testing
whether $\beta_{1}$ is equal to $\beta_{2}.$ For the moment, we consider the
case that all coefficients are subject to a possible break. In Section
\ref{Sec: Extend}, we consider the case that some of the coefficients are
known to be time invariant.

Let
\[
X_{1t}=X_{t}\cdot1\left\{  t\leq\left[  \lambda T\right]  \right\}  \text{,
}X_{2t}=X_{t}\cdot1\left\{  t\geq\left[  \lambda T\right]  +1\right\}  .
\]
Note that both $X_{1t}$ and $X_{2t}$ are nonstationary. The form of the
nonstationarity makes the problem at hand unique. Let $\beta=\left(  \beta
_{1}^{\prime},\beta_{2}^{\prime}\right)  ^{\prime}$ and $\tilde{X}_{t}%
=(X_{1t},X_{2t}).$ Then%
\[
Y_{t}=\tilde{X}_{t}\beta+u_{t},
\]
and the hypotheses of interest become $H_{0}:R\beta=0$ and $H_{1}:R\beta\neq0$
for $R=[\mathcal{R},-\mathcal{R}]\in\mathbb{R}^{p\times2m}.$

Denote $\tilde{X}=(\tilde{X}_{1}^{\prime},\ldots,\tilde{X}_{T}^{\prime
})^{\prime}$, $Y=\left(  Y_{1},\ldots,Y_{T}\right)  ^{\prime},$ and $u=\left(
u_{1},\ldots,u_{T}\right)  ^{\prime}.$ We estimate $\beta$ by OLS:%
\[
\hat{\beta}=(\tilde{X}^{\prime}\tilde{X})^{-1}\tilde{X}^{\prime}Y.
\]
The OLS estimator $\hat{\beta}$ satisfies%
\[
\sqrt{T}(\hat{\beta}-\beta)=\hat{Q}^{-1}\frac{1}{\sqrt{T}}\sum_{t=1}^{T}%
\tilde{X}_{t}^{\prime}u_{t},
\]
where
\[
\hat{Q}=\frac{\tilde{X}^{\prime}\tilde{X}}{T}=\left(
\begin{array}
[c]{cc}%
T^{-1}\sum_{t=1}^{[T\lambda]}X_{t}^{\prime}X_{t} & O\\
O & T^{-1}\sum_{T[\lambda]+1}^{T}X_{t}^{\prime}X_{t}%
\end{array}
\right)
\]
and $O$ is a matrix of zeros. To make inferences on $\beta$ such as testing
whether $R\beta$ is zero, we need to estimate the variance of $T^{-1/2}%
\sum_{t=1}^{T}\tilde{X}_{t}^{\prime}u_{t}.$ To this end, we first construct
the residual $\hat{u}_{t}=Y_{t}-\tilde{X}_{t}\hat{\beta},$ which serves as an
estimate for $u_{t}.$ Given a set of basis functions $\left\{  \phi_{j}\left(
\cdot\right)  \right\}  _{j=1}^{K},$ we then construct the series estimator of
the variance as
\[
\hat{\Omega}=\frac{1}{K}\sum_{j=1}^{K}\left[  \frac{1}{\sqrt{T}}\sum_{t=1}%
^{T}\phi_{j}\left(  \frac{t}{T}\right)  \tilde{X}_{t}^{\prime}\hat{u}%
_{t}\right]  ^{\otimes2},
\]
where, for a column vector $a,$ $a^{\otimes2}$ is the outer product of $a,$
that is, $a^{\otimes2}=aa^{\prime}.$ The asymptotic variance of $R\sqrt
{T}(\hat{\beta}-\beta)$ is then estimated by $R\hat{Q}^{-1}\hat{\Omega}\hat
{Q}^{-1}R^{\prime}.$

The Wald statistic for testing $H_{0}:R\beta=0$ against $H_{1}:R\beta\neq0$ is%
\[
F_{T}=T\cdot(R\hat{\beta})^{\prime}\left[  R\hat{Q}^{-1}\hat{\Omega}\hat
{Q}^{-1}R^{\prime}\right]  ^{-1}(R\hat{\beta}).
\]
When $p=1$ and we test $H_{0}:R\beta=0$ against a one-sided alternative, say,
$H_{1}:R\beta>0$, we can construct the t statistic:%
\[
t_{T}=\frac{\sqrt{T}\cdot R\hat{\beta}}{\left[  R\hat{Q}^{-1}\hat{\Omega}%
\hat{Q}^{-1}R^{\prime}\right]  ^{1/2}}.
\]
The forms of the F and t statistics are standard.

\section{Fixed-smoothing Asymptotic Distributions\label{FSA}}

To establish the asymptotic distributions of $F_{T}$ and $t_{T},$ we maintain
the following three assumptions:

\begin{assumption}
\label{Assumption XX}$T^{-1}\sum_{t=1}^{[Tr]}X_{t}^{\prime}X_{t}%
\rightarrow^{p}Q\cdot r$ uniformly over $r\in\lbrack0,1]$ and $Q$ is invertible.
\end{assumption}

\begin{assumption}
\label{Assumption XU}$T^{-1/2}\sum_{t=1}^{[Tr]}X_{t}^{\prime}u_{t}%
\rightarrow^{d}\Lambda W_{m}\left(  r\right)  $ for $r\in\lbrack0,1]$ where
$\Omega=\Lambda\Lambda^{\prime}$ is the long run variance of $\left\{
X_{t}^{\prime}u_{t}\right\}  $ and $W_{m}\left(  \cdot\right)  $ is an
$m\times1$ standard Brownian process.
\end{assumption}

\begin{assumption}
\label{Weight_function}The basis functions $\phi_{j}\left(  \cdot\right)  $,
$j=1,2,\ldots,K$ are piecewise monotonic and piecewise continuously differentiable.
\end{assumption}

\begin{lemma}
\label{Lemma_1}Let Assumptions \ref{Assumption XX} and \ref{Assumption XU}
hold. Then%
\[
\sqrt{T}(\hat{\beta}-\beta):=\left(
\begin{array}
[c]{c}%
\sqrt{T}(\hat{\beta}_{1}-\beta_{1})\\
\sqrt{T}(\hat{\beta}_{2}-\beta_{2})
\end{array}
\right)  \rightarrow^{d}\left(
\begin{array}
[c]{c}%
Q^{-1}\Lambda\cdot\frac{1}{\lambda}\int_{0}^{\lambda}dW_{m}\left(
\lambda\right) \\
Q^{-1}\Lambda\cdot\frac{1}{1-\lambda}\int_{\lambda}^{1}dW_{m}\left(
\lambda\right)
\end{array}
\right)  .
\]
If Assumption \ref{Weight_function} also holds, then
\[
\frac{1}{\sqrt{T}}\sum_{t=1}^{T}\phi_{j}\left(  \frac{t}{T}\right)  \tilde
{X}_{t}^{\prime}\hat{u}_{t}\rightarrow^{d}\left(
\begin{array}
[c]{c}%
\Lambda\cdot\int_{0}^{\lambda}\left[  \phi_{j}\left(  r\right)  -\bar{\phi
}_{j,1}\right]  dW_{m}\left(  r\right) \\
\Lambda\cdot\int_{\lambda}^{1}\left[  \phi_{j}\left(  r\right)  -\bar{\phi
}_{j,2}\right]  dW_{m}\left(  r\right)
\end{array}
\right)  ,
\]
where%
\[
\bar{\phi}_{j,1}=\frac{1}{\lambda}\int_{0}^{\lambda}\phi_{j}\left(  s\right)
ds\text{ and }\bar{\phi}_{j,2}=\frac{1}{1-\lambda}\int_{\lambda}^{1}\phi
_{j}\left(  s\right)  ds.
\]

\end{lemma}

Note that $\frac{1}{\lambda}\int_{0}^{\lambda}dW_{m}\left(  \lambda\right)  $
and $\frac{1}{1-\lambda}\int_{\lambda}^{1}dW_{m}\left(  \lambda\right)  $ are
the average changes of the Brownian motion over the intervals $\left[
0,\lambda\right]  $ and $\left[  \lambda,1\right]  ,$ respectively. Lemma
\ref{Lemma_1} shows that $\sqrt{T}(\hat{\beta}_{1}-\beta_{1})$ and $\sqrt
{T}(\hat{\beta}_{2}-\beta_{2})$ are (matrix) proportional to the average
changes.\ Given the independence of these changes over any non-overlapping
intervals, $\sqrt{T}(\hat{\beta}_{1}-\beta_{1})$ and $\sqrt{T}(\hat{\beta}%
_{2}-\beta_{2})$ are asymptotically independent.

Note that $\bar{\phi}_{j,1}$ can be regarded as an average of $\phi_{j}\left(
\cdot\right)  $ over the interval $\left[  0,\lambda\right]  .$ Similarly,
$\bar{\phi}_{j,2}$ can be regarded as an average of $\phi_{j}\left(
\cdot\right)  $ over the interval $\left[  \lambda,1\right]  .$ So $\phi
_{j}\left(  r\right)  -\bar{\phi}_{j,1}$ and $\phi_{j}\left(  r\right)
-\bar{\phi}_{j,2}$ are the demeaned versions of $\phi_{j}\left(  r\right)  $
over the intervals $\left[  0,\lambda\right]  $ and $\left[  \lambda,1\right]
,$ respectively.

Using Lemma \ref{Lemma_1}, we can prove our main theorem below.

\begin{theorem}
\label{Theorem_main1}Let Assumptions \ref{Assumption XX}%
--\ref{Weight_function} hold. Then, under the null hypothesis,
\begin{align}
F_{T}  &  \rightarrow^{d}\left[  \frac{1}{\lambda}\int_{0}^{\lambda}%
dW_{p}\left(  \lambda\right)  -\frac{1}{1-\lambda}\int_{\lambda}^{1}%
dW_{p}\left(  \lambda\right)  \right]  ^{\prime}\times\left[  \frac{1}{K}%
\sum_{j=1}^{K}\left\{  \int_{0}^{1}\tilde{\phi}_{j}\left(  r;\lambda\right)
dW_{p}\left(  r\right)  \right\}  ^{\otimes2}\right]  ^{-1}\nonumber\\
&  \times\left[  \frac{1}{\lambda}\int_{0}^{\lambda}dW_{p}\left(
\lambda\right)  -\frac{1}{1-\lambda}\int_{\lambda}^{1}dW_{p}\left(
\lambda\right)  \right] \nonumber\\
&  :=F_{\infty}, \label{F_limit}%
\end{align}
where
\begin{equation}
\tilde{\phi}_{j}\left(  r;\lambda\right)  =\frac{1}{\lambda}\left[  \phi
_{j}\left(  r\right)  -\bar{\phi}_{j,1}\right]  \cdot1\left\{  r\leq
\lambda\right\}  -\frac{1}{1-\lambda}\left[  \phi_{j}\left(  r\right)
-\bar{\phi}_{j,2}\right]  \cdot1\left\{  r>\lambda\right\}  .
\label{transformed_phi_fun}%
\end{equation}

When $p=1,$%
\[
t_{T}\rightarrow^{d}\left[  \frac{1}{\lambda}\int_{0}^{\lambda}dW_{p}\left(
\lambda\right)  -\frac{1}{1-\lambda}\int_{\lambda}^{1}dW_{p}\left(
\lambda\right)  \right]  \times\left[  \frac{1}{K}\sum_{j=1}^{K}\left\{
\int_{0}^{1}\tilde{\phi}_{j}\left(  r;\lambda\right)  dW_{p}\left(  r\right)
\right\}  ^{\otimes2}\right]  ^{-1/2}:=t_{\infty}.
\]

\end{theorem}

Like the finite sample distributions, the limiting distributions of $F_{T}$
and $t_{T}$ depend on $\lambda$ and the number and form of the basis
functions. This is an attractive feature of the fixed-smoothing
approximations, as they capture the effects of all these factors. More
importantly, the fixed-smoothing approximations capture the randomness of the
HAR variance estimator, which clearly affects the finite sample distributions
of $F_{T}$ and $t_{T}.$ This is why the fixed-smoothing asymptotic
approximations are more accurate than the chi-square or normal approximations.

\section{Asymptotic F and t Theory \label{Sec: F and t}}

The limiting distributions $F_{\infty}$ and $t_{\infty}$ in Theorem
\ref{Theorem_main1} are pivotal but nonstandard. We can approximate the
nonstandard distributions using a chi-square or t distribution. We can also
design a new set of basis functions so that $F_{\infty}$ and $t_{\infty}$
become the standard F and t distributions after some multiplicative adjustment.

\subsection{Chi-square and normal approximations}

Define
\[
\tilde{\phi}_{0}\left(  r;\lambda\right)  =\frac{1}{\lambda}1\left\{
r\leq\lambda\right\}  -\frac{1}{1-\lambda}1\left\{  r>\lambda\right\}  .
\]
Then
\begin{align*}
&  \frac{1}{\lambda}\int_{0}^{\lambda}dW_{p}\left(  \lambda\right)  -\frac
{1}{1-\lambda}\int_{\lambda}^{1}dW_{p}\left(  \lambda\right) \\
&  =\int_{0}^{1}\tilde{\phi}_{0}\left(  r;\lambda\right)  dW_{p}\left(
r\right)  \thicksim N\left(  0,\int_{0}^{1}\left[  \tilde{\phi}_{0}\left(
r;\lambda\right)  \right]  ^{2}dr\cdot I_{p}\right) \\
&  =N\left(  0,\frac{1}{\lambda\left(  1-\lambda\right)  }I_{p}\right)  ,
\end{align*}
and so%
\begin{align*}
\eta_{0}  &  :=\sqrt{\lambda\left(  1-\lambda\right)  }\left[  \frac
{1}{\lambda}\int_{0}^{\lambda}dW_{p}\left(  \lambda\right)  -\frac
{1}{1-\lambda}\int_{\lambda}^{1}dW_{p}\left(  \lambda\right)  \right] \\
&  =\sqrt{\lambda\left(  1-\lambda\right)  }\int_{0}^{1}\tilde{\phi}%
_{0}\left(  r;\lambda\right)  dW_{p}\left(  r\right)  \thicksim N\left(
0,I_{p}\right)  .
\end{align*}
As a result,
\[
\lambda\left(  1-\lambda\right)  F_{T}\rightarrow^{d}\eta_{0}^{\prime}\left(
\frac{1}{K}\sum_{j=1}^{K}\eta_{j}\eta_{j}^{\prime}\right)  ^{-1}\eta_{0}\text{
and }\sqrt{\lambda\left(  1-\lambda\right)  }t_{T}\rightarrow^{d}\eta
_{0}^{\prime}\left(  \frac{1}{K}\sum_{j=1}^{K}\eta_{j}\eta_{j}^{\prime
}\right)  ^{-1/2},
\]
where
\[
\eta_{j}=\int_{0}^{1}\tilde{\phi}_{j}\left(  r;\lambda\right)  dW_{p}\left(
r\right)  \text{ for }j=1,\ldots,K.
\]

When $K$ is relatively large, it is reasonable to approximate $K^{-1}%
\sum_{j=1}^{K}\eta_{j}\eta_{j}^{\prime}$ by its mean:
\[
E\left[  \frac{1}{K}\sum_{j=1}^{K}\eta_{j}\eta_{j}^{\prime}\right]
=I_{p}\cdot\frac{1}{K}\sum_{j=1}^{K}\int_{0}^{1}\left[  \tilde{\phi}%
_{j}\left(  r;\lambda\right)  \right]  ^{2}dr.
\]
With such an approximation, we have%
\begin{align*}
\lambda\left(  1-\lambda\right)  \cdot\frac{1}{K}\sum_{j=1}^{K}\int_{0}%
^{1}\left[  \tilde{\phi}_{j}\left(  r;\lambda\right)  \right]  ^{2}dr\cdot
F_{\infty}  &  \thicksim^{a}\chi_{p}^{2},\\
\sqrt{\lambda\left(  1-\lambda\right)  }\left[  \frac{1}{K}\sum_{j=1}^{K}%
\int_{0}^{1}\left[  \tilde{\phi}_{j}\left(  r;\lambda\right)  \right]
^{2}dr\right]  ^{1/2}\cdot t_{\infty}  &  \thicksim^{a}N(0,1),
\end{align*}
where `$\thicksim^{a}$' signifies distributional approximations. As a result,
we can employ the following approximations:
\begin{align}
F_{T}^{\ast}  &  :=\lambda\left(  1-\lambda\right)  \cdot\left\{  \frac{1}%
{KT}\sum_{j=1}^{K}\sum_{i=1}^{T}\left[  \tilde{\phi}_{j,T}\left(  \frac{i}%
{T};\lambda\right)  \right]  ^{2}\right\}  \cdot F_{T}\thicksim^{a}\chi
_{p}^{2},\label{F_star}\\
t_{T}^{\ast}  &  :=\sqrt{\lambda\left(  1-\lambda\right)  }\cdot\left\{
\frac{1}{KT}\sum_{j=1}^{K}\sum_{i=1}^{T}\left[  \tilde{\phi}_{j,T}\left(
\frac{i}{T};\lambda\right)  \right]  ^{2}\right\}  ^{1/2}\cdot t_{T}%
\thicksim^{a}N(0,1), \label{t_star}%
\end{align}
where $\tilde{\phi}_{j,T}\left(  r;\lambda\right)  $ is the finite sample
version of $\tilde{\phi}_{j}\left(  r;\lambda\right)  $ given by
\begin{align}
\tilde{\phi}_{j,T}\left(  r;\lambda\right)   &  =\frac{1}{\lambda}\left[
\phi_{j}\left(  r\right)  -\frac{1}{\left[  \lambda T\right]  }\sum
_{t=1}^{\left[  \lambda T\right]  }\phi_{j}\left(  \frac{t}{T}\right)
\right]  \cdot1\left\{  r\leq\lambda\right\} \nonumber\\
&  -\frac{1}{1-\lambda}\left[  \phi_{j}\left(  r\right)  -\frac{1}{T-\left[
\lambda T\right]  }\sum_{t=\left[  \lambda T\right]  +1}^{T}\phi_{j}\left(
\frac{t}{T}\right)  \right]  \cdot1\left\{  r>\lambda\right\}  .
\end{align}

It is important to point out that the chi-square and normal approximations are
not based on the original Wald and t statistics but rather on their modified
versions $F_{T}^{\ast}$ and $t_{T}^{\ast}$. To a great extent, the chi-square
and normal approximations we propose here improve upon the conventional
chi-square and normal approximations that are applied directly to the original
Wald and t statistics.

Note that the chi-square distribution and standard normal distribution in
(\ref{F_star}) and (\ref{t_star}) are not the asymptotic distributions of
$F_{T}^{\ast}$ and $t_{T}^{\ast}$ for a fixed $K.$ The fixed-$K$ asymptotic
distributions are given by
\begin{align}
F_{T}^{\ast}  &  \rightarrow^{d}\lambda\left(  1-\lambda\right)  \cdot\frac
{1}{K}\sum_{j=1}^{K}\int_{0}^{1}\left[  \tilde{\phi}_{j}\left(  r;\lambda
\right)  \right]  ^{2}dr\cdot F_{\infty}:=F_{\infty}^{\ast}%
,\label{F_star_nonstandard}\\
t_{T}^{\ast}  &  \rightarrow^{d}\sqrt{\lambda\left(  1-\lambda\right)  }%
\cdot\left\{  \frac{1}{K}\sum_{j=1}^{K}\int_{0}^{1}\left[  \tilde{\phi}%
_{j}\left(  r;\lambda\right)  \right]  \right\}  ^{1/2}\cdot t_{\infty
}:=t_{\infty}^{\ast}. \label{t_star_nonstandard}%
\end{align}
These follow directly from Theorem \ref{Theorem_main1}. The chi-square
distribution and standard normal distribution are only approximations to the
above nonstandard fixed-K asymptotic distributions.

\subsection{Asymptotic F and t Theory}

To obtain convenient fixed-K asymptotic approximations, we note that for each
$j=0,1,\ldots,K,$ $\eta_{j}$ is normal. For each $j\neq0,$ we have
\begin{align*}
&  cov\left(  \eta_{0},\eta_{j}\right) \\
&  =\int_{0}^{1}\tilde{\phi}_{0}\left(  r;\lambda\right)  \tilde{\phi}%
_{j}\left(  r;\lambda\right)  dr=\frac{1}{\lambda^{2}}\int_{0}^{1}\left[
\phi_{j}\left(  r\right)  -\bar{\phi}_{j,1}\right]  1\left\{  r\leq
\lambda\right\}  dr\\
&  +\frac{1}{\left(  1-\lambda\right)  ^{2}}\int_{0}^{1}\left[  \phi
_{j}\left(  r\right)  -\bar{\phi}_{j,2}\right]  1\left\{  r>\lambda\right\}
dr=0.
\end{align*}
So $\eta_{0}$ is independent of $\eta_{j}$, $j=1,\ldots,K.$ In addition,
\[
cov\left(  \eta_{j_{1}},\eta_{j_{2}}\right)  =\int_{0}^{1}\tilde{\phi}_{j_{1}%
}\left(  r;\lambda\right)  \tilde{\phi}_{j_{2}}\left(  r;\lambda\right)  dr.
\]
Therefore, if $\{\tilde{\phi}_{j}\left(  r;\lambda\right)  \}$ are
orthonormal, then $\eta_{j}$ for $j=0,1,\ldots,K$ are independent standard
normals. In this case, $\lambda\left(  1-\lambda\right)  F_{\infty}$ is a
quadratic form in a standard normal vector with an independent weighting
matrix. After some adjustment, we can show that $\lambda\left(  1-\lambda
\right)  F_{\infty}$ is equal to a standard F distribution and that $F_{T}$
converges to the F distribution. Similarly, $\sqrt{\lambda\left(
1-\lambda\right)  }\cdot t_{T}$ converges to Student's t distribution.

\begin{proposition}
\label{Prop_F}Let Assumptions \ref{Assumption XX}--\ref{Weight_function} hold.
If $\{\tilde{\phi}_{j}\left(  r;\lambda\right)  \}$ are orthonormal, then
\[
\tilde{F}_{T}^{\ast}:=\frac{K-p+1}{Kp}\cdot\lambda\left(  1-\lambda\right)
\cdot F_{T}\rightarrow^{d}F_{p,K-p+1},
\]
and%
\[
\tilde{t}_{T}^{\ast}:=\sqrt{\lambda\left(  1-\lambda\right)  }\cdot
t_{T}\rightarrow^{d}t_{K}%
\]
where $F_{p,K-p+1}$ is the standard F distribution with the degrees of freedom
$\left(  p,K-p+1\right)  $ and $t_{K}$ is Student's t distribution with
degrees of freedom $K.$
\end{proposition}

This is a very convenient result, as the fixed-smoothing asymptotic
approximations are standard distributions and there is no need to simulate
critical values.

When $\{\tilde{\phi}_{j}\left(  r;\lambda\right)  \}$ are orthonormal, we have
$K^{-1}\sum_{j=1}^{K}\int_{0}^{1}\left[  \tilde{\phi}_{j}\left(
r;\lambda\right)  \right]  ^{2}dr=1.$ In view of this, we can see that the
definitions of $\tilde{F}_{T}^{\ast}$ and $\tilde{t}_{T}^{\ast}$ are similar
to those of $F_{T}^{\ast}$ and $t_{T}^{\ast}$ given in
(\ref{F_star_nonstandard}) and (\ref{t_star_nonstandard}). The only difference
is that there is an additional degrees-of-freedom-adjustment factor in
$\tilde{F}_{T}^{\ast}$ when $p>1.$

\subsection{Designing the bases\label{bases_design}}

To design the basis functions such that $\{\tilde{\phi}_{j}\left(
r;\lambda\right)  \}$ are orthonormal, we need the following lemma.

\begin{lemma}
\label{Lemma: covariance}Let $\delta\left(  \cdot\right)  $ be the Dirac delta
function such that
\[
\int_{0}^{1}\int_{0}^{1}\phi_{j_{1}}\left(  r\right)  \delta\left(
r-s\right)  \phi_{j_{2}}\left(  s\right)  drds=\int_{0}^{1}\phi_{j_{1}}\left(
r\right)  \phi_{j_{2}}\left(  r\right)  dr.
\]
Then%
\[
\int_{0}^{1}\tilde{\phi}_{j_{1}}\left(  r;\lambda\right)  \tilde{\phi}_{j_{2}%
}\left(  r;\lambda\right)  dr=\int_{0}^{1}\int_{0}^{1}C(r,s;\lambda
)\phi_{j_{1}}\left(  r\right)  \phi_{j_{2}}\left(  s\right)  drds,
\]
where%
\begin{align*}
C(r,s;\lambda)  &  =\left[  \delta(r-s)-\frac{1}{\lambda}\right]
\frac{1\left\{  \left(  r,s\right)  \in\lbrack0,\lambda]\times\lbrack
0,\lambda]\right\}  }{\lambda^{2}}\\
&  +\left[  \delta(r-s)-\frac{1}{1-\lambda}\right]  \frac{1\left\{  \left(
r,s\right)  \in\lbrack\lambda,1]\times\lbrack\lambda,1]\right\}  }{\left(
1-\lambda\right)  ^{2}}.
\end{align*}

\end{lemma}

Let%
\begin{align*}
W_{p}\left(  r;\lambda\right)   &  =\frac{1}{\lambda}\left[  W_{p}\left(
r\right)  -\frac{r}{\lambda}W_{p}\left(  \lambda\right)  \right]
\cdot1\left\{  0\leq r\leq\lambda\right\} \\
&  -\frac{1}{1-\lambda}\left\{  W_{p}\left(  r\right)  -W_{p}\left(
\lambda\right)  -\frac{r-\lambda}{1-\lambda}\left[  W_{p}\left(  1\right)
-W_{p}\left(  \lambda\right)  \right]  \right\}  \cdot1\left\{  \lambda
<r\leq1\right\}
\end{align*}
be the transformed Brownian motion. Then we have%
\[
\int_{0}^{1}\tilde{\phi}_{j_{1}}\left(  r;\lambda\right)  dW_{p}\left(
r\right)  =\int_{0}^{1}\phi_{j}\left(  r\right)  dW_{p}(r;\lambda),
\]
and
\[
E\left[  dW_{p}\left(  r;\lambda\right)  dW_{p}^{\prime}\left(  s;\lambda
\right)  \right]  =I_{p}\cdot C(r,s;\lambda)drds.
\]
Therefore, $C(r,s;\lambda)$ can be regarded as the covariance kernel function
for the transformed Brownian motion.

To design the basis functions $\left\{  \phi_{j}\left(  r\right)  \right\}  $
such that $\{\tilde{\phi}_{j}\left(  r;\lambda\right)  \}$ are orthonormal on
$L^{2}[0,1],$ we require that $\left\{  \phi_{j}\left(  r\right)  \right\}  $
be orthonormal with respect to the covariance kernel function $C(r,s;\lambda
),$ that is,
\begin{equation}
\int_{0}^{1}\int_{0}^{1}C(r,s;\lambda)\phi_{j_{1}}\left(  r\right)
\phi_{j_{2}}\left(  s\right)  drds=1\left\{  j_{1}=j_{2}\right\}  . \label{OC}%
\end{equation}
This can be achieved by applying the Gram--Schmidt orthonormalization to any
set of basis functions on $L^{2}[0,1]$. The chart below illustrates the procedure:

\begin{center}%
\begin{tabular}
[c]{|lll|}\hline
$\left\{  \phi_{j}\right\}  $ & $\underset{\text{{\footnotesize estimation
error}}}{\rightarrow}$ & $\left\{  \tilde{\phi}_{j}\right\}  $ $(\text{may not
be orthonormal on }L^{2}\left[  0,1\right]  )$\\
$\ $\ $\downarrow${\tiny GS} &  & \\
$\{\phi_{j}^{\ast}\}$ & $\underset{\text{{\footnotesize estimation error}%
}}{\rightarrow}$ & $\left\{  \tilde{\phi}_{j}^{\ast}\right\}  $
$(\text{orthonormal on }L^{2}\left[  0,1\right]  )$\\\hline
\end{tabular}

\end{center}

In the above, $\left\{  \phi_{j}\right\}  $ is the initial set of basis
functions, and $\{\phi_{j}^{\ast}\}$ is the Gram-Schmidt orthonormalized set.
\textquotedblleft$\phi_{j}\rightarrow\tilde{\phi}_{j}$\textquotedblright\ and
\textquotedblleft$\phi_{j}^{\ast}\rightarrow\tilde{\phi}_{j}^{\ast}%
$\textquotedblright\ reflect the effect of the estimation error in estimating
$\beta:$ had we known $\beta,$ we would have used the true $u_{t}$ instead of
$\hat{u}_{t}$ in constructing the variance estimator, and the key elements of
the weighting matrix in (\ref{F_limit}) in Theorem \ref{Theorem_main1} would
have been $\int_{0}^{1}\phi_{j}\left(  r\right)  dW_{p}\left(  r\right)  $
instead of $\int_{0}^{1}\tilde{\phi}_{j}\left(  r;\lambda\right)
dW_{p}\left(  r\right)  .$ The Gram-Schmidt orthonormalization ensures that
$\{\phi_{j}^{\ast}\}$ are orthonormal with respect to the covariance kernel
$C(r,s;\lambda):\int_{0}^{1}\int_{0}^{1}\phi_{j_{1}}^{\ast}\left(  r\right)
\phi_{j_{2}}^{\ast}\left(  s\right)  C(r,s;\lambda)drds=1\left\{  j_{1}%
=j_{2}\right\}  .$ In view of%
\[
\int_{0}^{1}\int_{0}^{1}\phi_{j_{1}}^{\ast}\left(  r\right)  \phi_{j_{2}%
}^{\ast}\left(  s\right)  C(r,s;\lambda)drds=\int_{0}^{1}\tilde{\phi}_{j_{1}%
}^{\ast}\left(  r\right)  \tilde{\phi}_{j_{2}}^{\ast}\left(  r\right)  dr,
\]
we have: $\{\tilde{\phi}_{j}^{\ast}\}$ are orthonormal on $L^{2}\left[
0,1\right]  .$

If we use $\{\phi_{j}^{\ast}\}$ in constructing the variance estimator, then
\[
\lambda\left(  1-\lambda\right)  F_{T}\rightarrow^{d}\eta_{0}^{\prime}\left(
\frac{1}{K}\sum_{j=1}^{K}\eta_{j}\eta_{j}^{\prime}\right)  ^{-1}\eta_{0}%
\]
for $\eta_{j}=\int_{0}^{1}\tilde{\phi}_{j}^{\ast}\left(  r;\lambda\right)
dW_{p}\left(  r\right)  \thicksim iidN(0,I_{p})$ because $\{\tilde{\phi}%
_{j}^{\ast}\}$ are orthonormal on $L^{2}\left[  0,1\right]  .$ Moreover, for
$j=1,\ldots,K,$ $\eta_{j}$ is independent of $\eta_{0}.$ Therefore, the
asymptotic F theory in Proposition \ref{Prop_F} holds. Similarly, the
asymptotic t theory holds.

Instead of searching for the basis functions that satisfy (\ref{OC}), we
search for their discrete versions: the basis vectors. For each basis function
$\phi_{k}\left(  r\right)  ,$ the corresponding basis vector is defined as
\[
\boldsymbol{\phi}_{k}=\left(  \phi_{k}\left(  \frac{1}{T}\right)  ,\phi
_{k}\left(  \frac{2}{T}\right)  ,\ldots,\phi_{k}\left(  \frac{T}{T}\right)
\right)  ^{\prime}.
\]
Let $\mathbf{C}_{T}:=\mathbf{C}_{T}\left(  \lambda\right)  $ be the $T\times
T$ matrix whose $\left(  i,j\right)  $-th element is equal to
\begin{align*}
&  C_{T}\left(  i,j;\lambda\right)  =\left[  T\cdot1\left\{  i=j\right\}
-\frac{1}{\lambda}\right]  \frac{1\left\{  \left(  \frac{i}{T},\frac{j}%
{T}\right)  \in\lbrack0,\lambda]\times\lbrack0,\lambda]\right\}  }{\lambda
^{2}}\\
&  +\left[  T\cdot1\left\{  i=j\right\}  -\frac{1}{1-\lambda}\right]
\frac{1\left\{  \left(  \frac{i}{T},\frac{j}{T}\right)  \in(\lambda
,1]\times(\lambda,1]\right\}  }{\left(  1-\lambda\right)  ^{2}}.
\end{align*}
By definition, $\mathbf{C}_{T}$ is symmetric and positive-definite. It is the
discrete version of $C(r,s;\lambda).$ For any two vectors $r_{1},r_{2}%
\in\mathbb{R}^{T}$, we define the inner product%
\begin{equation}
\left\langle r_{1},r_{2}\right\rangle =r_{1}^{\prime}\mathbf{C}_{T}r_{2}%
/T^{2}. \label{inner product}%
\end{equation}
Then the discrete analogue of (\ref{OC}) is%
\begin{equation}
\left\langle \boldsymbol{\phi}_{j_{1}},\boldsymbol{\phi}_{j_{2}}\right\rangle
=1\left\{  j_{1}=j_{2}\right\}  \text{ for }j_{1},j_{2}=1,\ldots,K.
\label{OC3}%
\end{equation}

Given any basis vectors $\boldsymbol{\phi}_{1},\ldots,\boldsymbol{\phi}_{K},$
we now apply the Gram--Schmidt orthonormalization via the Cholesky
decomposition. Let $\boldsymbol{\phi}=(\boldsymbol{\phi}_{1},\mathbf{\ldots
},\boldsymbol{\phi}_{K})$ be the $T\times K$ matrix of basis vectors. Let
$U_{T}\in\mathbb{R}^{K\times K}$\ be the upper triangular factor in the
Cholesky decomposition of $\boldsymbol{\phi}^{\prime}\mathbf{C}_{T}%
\boldsymbol{\phi}/T^{2}$ such that $\boldsymbol{\phi}^{\prime}\mathbf{C}%
_{T}\boldsymbol{\phi}/T^{2}=U_{T}^{\prime}U_{T}.$ Define
\[
\boldsymbol{\phi}^{\ast}=\boldsymbol{\phi}U_{T}^{-1}:=(\boldsymbol{\phi}%
_{1}^{\ast},\mathbf{\ldots},\boldsymbol{\phi}_{K}^{\ast}).
\]
We then have
\[
(\boldsymbol{\phi}^{\ast})^{\prime}\mathbf{C}_{T}\boldsymbol{\phi}^{\ast
}/T^{2}=\left(  U_{T}^{\prime}\right)  ^{-1}\boldsymbol{\phi}^{\prime
}\mathbf{C}_{T}\boldsymbol{\phi}U_{T}^{-1}/T^{2}=\left(  U_{T}^{\prime
}\right)  ^{-1}U_{T}^{\prime}U_{T}U_{T}^{-1}=I_{K}.
\]
That is, the columns of the matrix $\boldsymbol{\phi}^{\ast}$ satisfy the
conditions in (\ref{OC3}).

Note that the $\left(  k_{1},k_{2}\right)  $-th element of $\boldsymbol{\phi
}^{\prime}\mathbf{C}_{T}\boldsymbol{\phi}/T^{2}$ satisfies
\begin{align*}
&  \frac{1}{T^{2}}\sum_{j=1}^{T}\sum_{i=1}^{T}\phi_{k_{1}}\left(  \frac{i}%
{T}\right)  C_{T}\left(  i,j;\lambda\right)  \phi_{k_{2}}\left(  \frac{j}%
{T}\right) \\
&  \rightarrow\int_{0}^{1}\int_{0}^{1}C\left(  r,s;\lambda\right)  \phi
_{k_{1}}\left(  r\right)  \phi_{k_{2}}\left(  s\right)  drds=\int_{0}%
^{1}\tilde{\phi}_{k_{1}}\left(  r;\lambda\right)  \tilde{\phi}_{k_{2}}\left(
r;\lambda\right)  dr\\
&  =cov(\eta_{k_{1}},\eta_{k_{2}})\text{ as }T\rightarrow\infty.
\end{align*}
This implies that $U_{T}$ converges to the upper triangular factor of the
Cholesky decomposition of \textrm{var}$(\eta_{1},\ldots,\eta_{K}).$ As a
result, every transformed basis vector is approximately equal to a linear
combination of the original basis vectors. The implied basis functions are
thus equal to linear combinations of the original basis functions. Therefore,
if Assumption \ref{Weight_function} holds for the original basis functions, it
also holds for the transformed basis functions. It then follows that
Proposition \ref{Prop_F} holds when $\left\{  \boldsymbol{\phi}_{1}^{\ast
},\mathbf{\ldots},\boldsymbol{\phi}_{K}^{\ast}\right\}  $ are used as the
basis vectors in constructing the asymptotic variance estimator. More
specifically, if we estimate $\Omega$ by
\[
\hat{\Omega}=\frac{1}{K}\sum_{j=1}^{K}\left[  \frac{1}{\sqrt{T}}\sum_{t=1}%
^{T}\phi_{j,t}^{\ast}\tilde{X}_{t}^{\prime}\hat{u}_{t}\right]  ^{\otimes2},
\]
where $\phi_{j,t}^{\ast}$ is the $t$-th element of the vector
$\boldsymbol{\phi}_{j}^{\ast},$ then the asymptotic F and t results in
Proposition \ref{Prop_F} hold.

\section{The Chow Test in the presence of time-invariant effects
\label{Sec: Extend}}

Suppose there is another covariate vector $Z_{t}\in\mathbb{R}^{\ell}$ whose
effect on $Y_{t}$ does not change over time so that we have the model:
\[
Y_{t}=X_{t}\cdot1\left\{  t\leq\left[  \lambda T\right]  \right\}  \cdot
\beta_{1}+X_{t}\cdot1\left\{  t>\left[  \lambda T\right]  +1\right\}
\cdot\beta_{2}+Z_{t}\gamma+u_{t}.
\]
Let $Z=(Z_{1}^{\prime},\ldots,Z_{T}^{\prime})^{\prime}$ and $M_{Z}%
=I_{T}-Z(Z^{\prime}Z)^{-1}Z^{\prime}.$ Then%
\[
M_{Z}Y=M_{Z}\tilde{X}\beta+M_{Z}u.
\]
The OLS estimator of $\beta=\left(  \beta_{1}^{\prime},\beta_{2}^{\prime
}\right)  ^{\prime}$ is now%
\[
\hat{\beta}=(\tilde{X}^{\prime}M_{Z}\tilde{X})^{-1}\tilde{X}^{\prime}M_{Z}Y.
\]
Let $\hat{u}=\left(  \hat{u}_{1},\ldots,\hat{u}_{T}\right)  ^{\prime}%
=M_{Z}Y-M_{Z}\tilde{X}\hat{\beta}=M_{Z}u-M_{Z}\tilde{X}(\hat{\beta}-\beta)$
and $\tilde{X}_{z}=(\tilde{X}_{z,1}^{\prime},\ldots,\tilde{X}_{z,T}^{\prime
})^{\prime}=M_{Z}\tilde{X}.$ Define
\[
\hat{Q}_{\tilde{X}\cdot Z}=\frac{\tilde{X}^{\prime}M_{Z}\tilde{X}}{T}\text{
and }\hat{\Omega}=\frac{1}{K}\sum_{j=1}^{K}\left[  \frac{1}{\sqrt{T}}%
\sum_{t=1}^{T}\phi_{j}\left(  \frac{t}{T}\right)  \tilde{X}_{z,t}^{\prime}%
\hat{u}_{t}\right]  ^{\otimes2}.
\]
The Wald statistic for testing $H_{0}:R\beta=0$ against $H_{1}:R\beta\neq0$
takes the same form as before:
\[
F_{T}=T\cdot(R\hat{\beta})^{\prime}\left[  R\hat{Q}^{-1}\hat{\Omega}\hat
{Q}^{-1}R^{\prime}\right]  ^{-1}(R\hat{\beta}).
\]
When $p=1,$ we construct the t statistic:
\[
t_{T}=\frac{\sqrt{T}\cdot R\hat{\beta}}{\left[  R\hat{Q}^{-1}\hat{\Omega}%
\hat{Q}^{-1}R^{\prime}\right]  ^{1/2}}.
\]

To establish the asymptotic distributions of $F_{T}$ and $t_{T},$ we maintain
the two assumptions below, which are analogous to Assumptions
\ref{Assumption XX} and \ref{Assumption XU}.

\begin{assumption}
\label{Assumption WW}$T^{-1}\sum_{t=1}^{[Tr]}\left(  X_{t},Z_{t}\right)
^{\prime}\left(  X_{t},Z_{t}\right)  \rightarrow^{p}Q\cdot r$ uniformly over
$r\in\lbrack0,1]$ for a $\left(  m+\ell\right)  \times\left(  m+\ell\right)  $
invertible matrix $Q$.
\end{assumption}

\begin{assumption}
\label{Assumption WU}$T^{-1/2}\sum_{t=1}^{[Tr]}\left(  X_{t},Z_{t}\right)
^{\prime}u_{t}\rightarrow^{d}\Lambda W_{m+\ell}\left(  r\right)  $ for
$r\in\lbrack0,1]$ where $\Lambda\Lambda^{\prime}$ is the long run variance of
the process $\{\left(  X_{t},Z_{t}\right)  ^{\prime}u_{t}\}$ and $W_{m+\ell
}\left(  \cdot\right)  $ is an $\left(  \ell+m\right)  \times1$ standard
Brownian process.
\end{assumption}

We partition $Q$ and $\Lambda$ according to
\[
Q=\left(
\begin{array}
[c]{cc}%
Q_{XX} & Q_{XZ}\\
Q_{ZX} & Q_{ZZ}%
\end{array}
\right)  \text{ and }\Lambda=\left(
\begin{array}
[c]{c}%
\Lambda_{X}\\
\Lambda_{Z}%
\end{array}
\right)  ,
\]
where $Q_{XX}\in\mathbb{R}^{m\times m},Q_{ZZ}\in\mathbb{R}^{\ell\times\ell},$
$\Lambda_{X}\in\mathbb{R}^{m\times(\ell+m)},$ and $\Lambda_{Z}\in
\mathbb{R}^{\ell\times(\ell+m)}.$

\begin{theorem}
\label{Theorem:General_Case}Let Assumptions \ref{Weight_function},
\ref{Assumption WW}, and \ref{Assumption WU} hold. Then

$\left(  a\right)  $%
\[
R\sqrt{T}(\hat{\beta}-\beta)\rightarrow^{d}\mathcal{R}Q_{XX}^{-1}\Lambda
_{X}\left(  \frac{1}{\lambda}\int_{0}^{\lambda}dW_{m+\ell}\left(
\lambda\right)  -\frac{1}{1-\lambda}\int_{\lambda}^{1}dW_{m+\ell}\left(
\lambda\right)  \right)  .
\]

$\left(  b\right)  $%
\[
R\hat{Q}_{\tilde{X}\cdot Z}^{-1}\frac{1}{\sqrt{T}}\sum_{t=1}^{T}\phi
_{j}\left(  \frac{t}{T}\right)  X_{z,t}^{\prime}\hat{u}_{t}\rightarrow
^{d}\mathcal{R}Q_{XX}^{-1}\Lambda_{X}\int_{0}^{1}\tilde{\phi}_{j}\left(
r;\lambda\right)  dW_{m+\ell}\left(  r\right)
\]
jointly over $j=1,2,...,K.$

$\left(  c\right)  $%
\begin{align}
F_{T}  &  \rightarrow^{d}\left[  \frac{1}{\lambda}\int_{0}^{\lambda}%
dW_{p}\left(  \lambda\right)  -\frac{1}{1-\lambda}\int_{\lambda}^{1}%
dW_{p}\left(  \lambda\right)  \right]  ^{\prime}\times\left[  \frac{1}{K}%
\sum_{j=1}^{K}\left\{  \int_{0}^{1}\tilde{\phi}_{j}\left(  r;\lambda\right)
dW_{p}\left(  r\right)  \right\}  ^{\otimes2}\right]  ^{-1}\nonumber\\
&  \times\left[  \frac{1}{\lambda}\int_{0}^{\lambda}dW_{p}\left(
\lambda\right)  -\frac{1}{1-\lambda}\int_{\lambda}^{1}dW_{p}\left(
\lambda\right)  \right]  .\nonumber
\end{align}

When $p=1,$%
\[
t_{T}\rightarrow^{d}\left[  \frac{1}{\lambda}\int_{0}^{\lambda}dW_{p}\left(
\lambda\right)  -\frac{1}{1-\lambda}\int_{\lambda}^{1}dW_{p}\left(
\lambda\right)  \right]  \times\left[  \frac{1}{K}\sum_{j=1}^{K}\left\{
\int_{0}^{1}\tilde{\phi}_{j}\left(  r;\lambda\right)  dW_{p}\left(  r\right)
\right\}  ^{\otimes2}\right]  ^{-1/2}.
\]

\end{theorem}

Theorem \ref{Theorem:General_Case} shows that the limiting distributions of
the Wald statistic and t statistic are the same as in the case without the
extra covariate $Z_{t}.$ The asymptotic F and t limit theory can be developed
in exactly the same way as in Section \ref{Sec: F and t}. We present the
result formally as a corollary.

\begin{corollary}
Let Assumptions \ref{Weight_function}, \ref{Assumption WW}, and
\ref{Assumption WU} hold. Suppose that the Gram--Schmidt transformed basis
vectors $\boldsymbol{\phi}_{1}^{\ast},...,\boldsymbol{\phi}_{K}^{\ast}$ are
used in constructing the variance estimator, that is,
\[
\hat{\Omega}=\frac{1}{K}\sum_{j=1}^{K}\left[  \frac{1}{\sqrt{T}}\sum_{t=1}%
^{T}\phi_{j,t}^{\ast}\tilde{X}_{z,t}^{\prime}\hat{u}_{t}\right]  ^{\otimes2}%
\]
where $\phi_{j,t}^{\ast}$ is the $t$-th element of the vector
$\boldsymbol{\phi}_{j}^{\ast}.$ Then
\[
\tilde{F}_{T}^{\ast}:=\frac{K-p+1}{Kp}\cdot\lambda\left(  1-\lambda\right)
\cdot F_{T}\rightarrow^{d}F_{p,K-p+1},
\]
and%
\[
\tilde{t}_{T}^{\ast}:=\sqrt{\lambda\left(  1-\lambda\right)  }\cdot
t_{T}\rightarrow^{d}t_{K}.
\]

\end{corollary}

\section{Simulation\ Evidence\label{Sec: Simulation}}

In this section, we investigate the finite sample properties of the proposed F
test. We consider the linear regression model with $m=2$ and $X_{t}%
=(1,q_{t}).$ The regressor $q_{t}$ follows an AR(1) process, and the error
$u_{t}$ follows an independent AR(1) or ARMA(1,1) process. That is,
\begin{align*}
q_{t}  &  =\rho q_{t-1}+\epsilon_{q,t}\\
u_{t}  &  =\rho u_{t-1}+\epsilon_{u,t}+\psi\epsilon_{u,t-1}%
\end{align*}
where\ both $\epsilon_{q,t}$ and $\epsilon_{u,t}$ are iid $N(0,1)$ and
$\left\{  \epsilon_{q,t},t=1,\ldots,T\right\}  $ are independent of $\left\{
\epsilon_{u,t}:t=1,2,\ldots,T\right\}  .$ Note that the AR parameter $\rho$ is
the same for $q_{t}$ and $u_{t}.$

We consider the sample sizes $T=100,200$, and $500.$ We let $\lambda=0.4.$
Without the loss of generality, we set $\beta_{1}=(0,0)^{\prime}$ and
$\beta_{2}=(0,0)^{\prime}$ under the null. We consider testing $H_{0}%
:\beta_{1}=\beta_{2}$ against $H_{1}:\beta_{1}\neq\beta_{2}$ so that $p=2$.

We consider two pairs of different tests, both of which are based on the
series variance estimators. The first pair uses the (usual) Fourier bases
\begin{equation}
\left\{  \phi_{2j-1}\left(  r\right)  =\sqrt{2}\cos\left(  2j\pi r\right)
,\text{ }\phi_{2j}=\sqrt{2}\sin\left(  2j\pi r\right)  ,j=1,\ldots
,K/2\right\}  . \label{cosine-sine}%
\end{equation}
Each test in this pair is based on the same test statistic $F_{T}^{\ast}$
defined in (\ref{F_star}) but uses different reference distributions. The
first test uses the chi-square approximation ($\chi_{2}^{2}$) while the second
test uses the nonstandard fixed-smoothing approximation given in
(\ref{F_star_nonstandard}). We refer to the two tests as \textquotedblleft%
$\chi^{2}:$ Fourier Bases\textquotedblright\ and \textquotedblleft$F_{\infty
}^{\ast}:$ Fourier Bases,\textquotedblright\ respectively. The nonstandard
critical values are simulated. We approximate the standard Brownian motion in
the nonstandard distribution using scaled partial sums of 1000 iid $N(0,1)$
random variables. To compute the nonstandard critical values, we use 10,000
simulation replications.

The second pair of tests uses the transformed Fourier bases via the
Gram--Schmidt orthogonalization given in Section \ref{bases_design}. Each of
the two tests in this pair is based on the same test statistic $\tilde{F}%
_{T}^{\ast}$ defined in Proposition \ref{Prop_F}. The first test uses the
standard F approximation, and the second test uses the rescaled chi-square
distribution $Kp\left[  K-p+1\right]  ^{-1}\chi_{2}^{2}.$ Equivalently, the
second test in this pair employs the test statistic $\tilde{F}_{T}%
=\lambda\left(  1-\lambda\right)  \cdot F_{T}$ and the standard chi-square
approximation ($\chi_{2}^{2}$). We refer to the two tests as \textquotedblleft%
$\chi^{2}:$ Transformed Bases\textquotedblright\ and \textquotedblleft$F:$
Transformed Bases,\textquotedblright\ respectively. The chi-square test in the
second pair is used to illustrate the effectiveness of the F approximation in
reducing the size distortion.

The nominal level of all tests is $5\%.$ The number of simulation replications
is 10,000. Figures \ref{figure_ar_T100} and \ref{figure_ar_T500} report the
null rejection probability for each test for the sample sizes $T=100$ and
$T=500$ when $q_{t}$ and $u_{t}$ follow independent AR(1) processes with the
same AR parameter $\rho.$ Several patterns emerge from these two figures:

\begin{itemize}
\item Regardless of the bases used, the chi-square tests over-reject the null
by a large margin, especially when $K$ is small.

\item Regardless of the bases used, the nonstandard test and F test are much
more accurate than the chi-square tests.

\item For each given value of $K,$ the null rejection probabilities of the
nonstandard test and F test are close to each other. This shows that, in terms
of size accuracy, using the F approximation (when the transformed Fourier
bases are employed) is as good as using the nonstandard approximation (when
the Fourier bases are employed). However, the F approximation is more
convenient to use and, hence, is preferred.

\item For each given value of $K,$ the null rejection probabilities of the two
chi-square tests are close to each other, although the one based on the
transformed Fourier bases is somewhat more accurate. This shows that the bases
do not have a large effect on the quality of the chi-square approximation.

\item The nonstandard test and standard F test can still have quite some size
distortion if $K$ is large and the regressor and error processes are
persistent. The size distortion comes from the bias of the variance estimator.
When $K$ is large, we take an average over a frequency window that is too
large when the processes are highly persistent, that is, when the spectral
density of $\left\{  x_{t}u_{t}\right\}  $ is not very flat at the origin. So,
it is important to use a data-driven $K$ to obtain an accurate test in practice.

\item Comparing the two figures, we see that the size distortion of every test
becomes smaller when the sample\ size is larger.
\end{itemize}

%

\begin{figure}[ptbh]%
\centering
\includegraphics[
height=6.826in,
width=6.461in
]%
{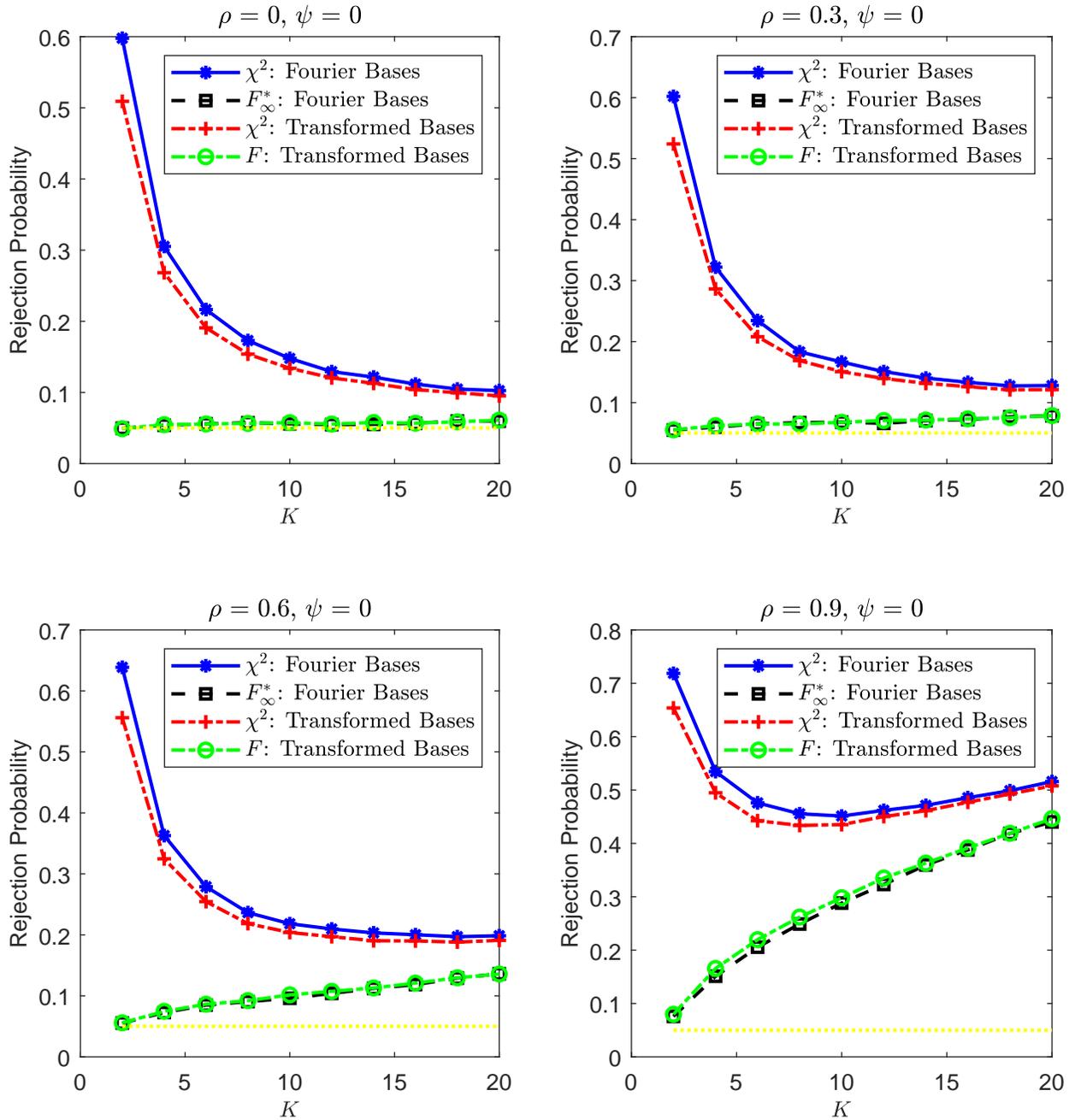}%
\caption{The empirical null rejection probabilities of different 5\% tests
when $T=100$ for a range of different K values from 2 to 20 with increment 2.
}%
\label{figure_ar_T100}%
\end{figure}
%

\begin{figure}[ptbh]%
\centering
\includegraphics[
height=6.826in,
width=6.461in
]%
{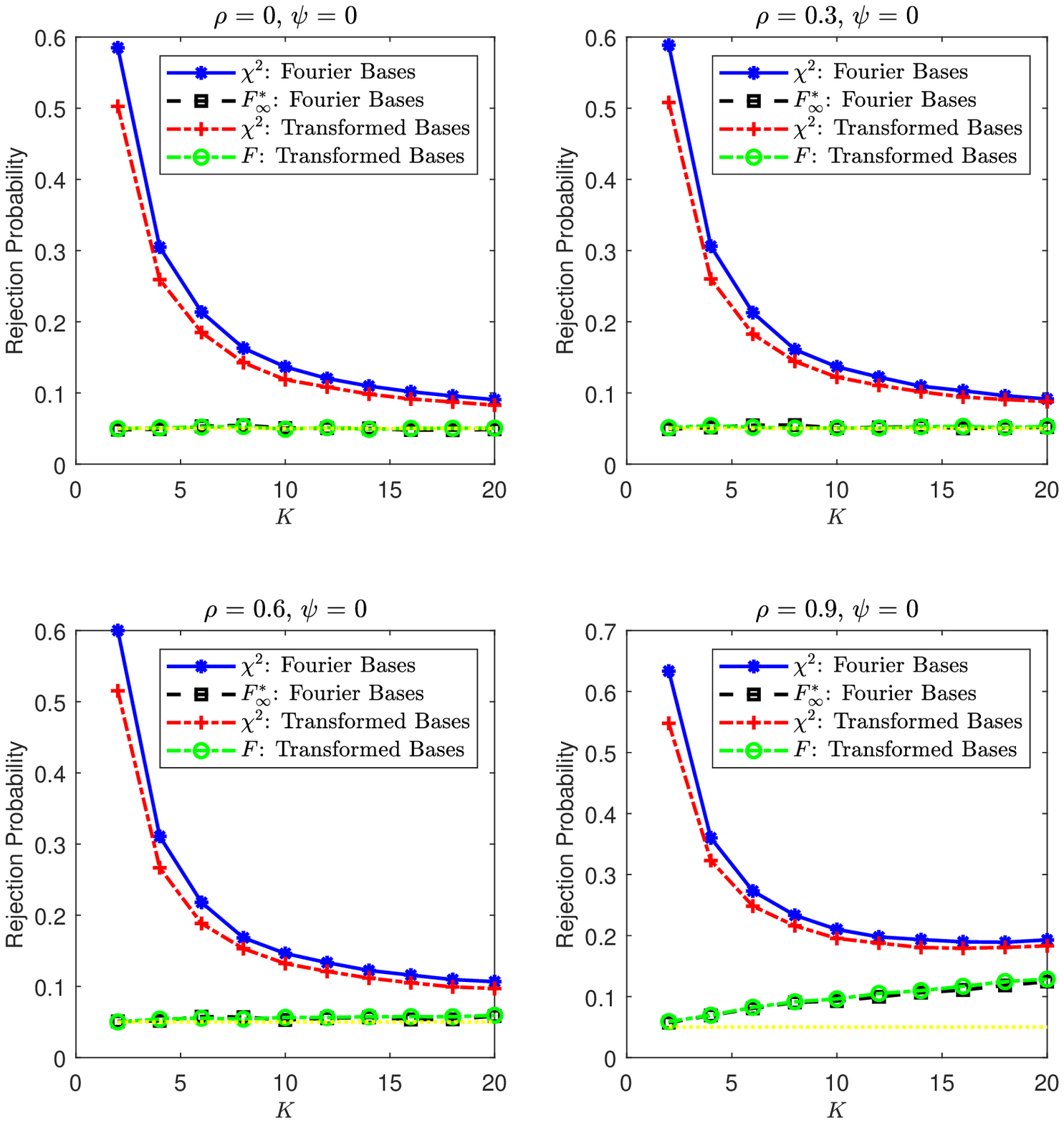}%
\caption{The empirical null rejection probabilities of different 5\% tests
when $T=500$ for a range of different K values from 2 to 20 with increment 2.
}%
\label{figure_ar_T500}%
\end{figure}

Figure \ref{Figure_ARMAT200} reports the null rejection probabilities when the
sample size $T$ is $200$ and when the error process may have an MA component
and the AR parameter may be negative. As in Figures \ref{figure_ar_T100} and
\ref{figure_ar_T500}, the same patterns emerge.%

\begin{figure}[ptbh]%
\centering
\includegraphics[
height=6.826in,
width=6.461in
]%
{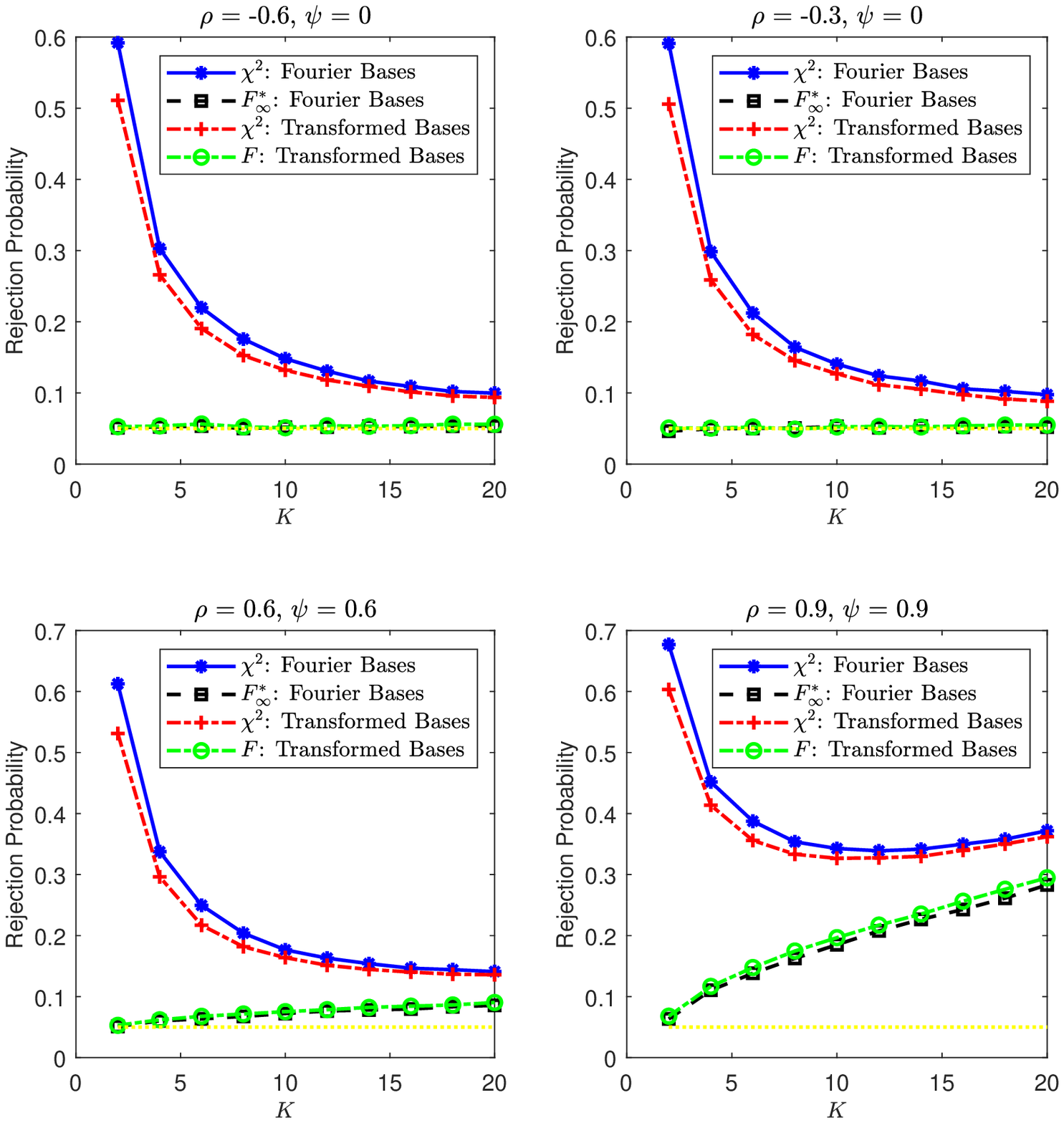}%
\caption{The empirical null rejection probabilities of different 5\% tests
when $T=200$ for a range of different K values from 2 to 20 with increment 2.
}%
\label{Figure_ARMAT200}%
\end{figure}

Next, we consider the size properties of the tests with a data-driven $K.$
Note that%
\[
R\sqrt{T}(\hat{\beta}-\beta)=\frac{1}{\sqrt{T}}\sum_{t=1}^{T}R\left(
\frac{\tilde{X}^{\prime}\tilde{X}}{T}\right)  ^{-1}\tilde{X}_{t}^{\prime}%
u_{t}=\frac{1}{\sqrt{T}}\sum_{t=1}^{T}v_{t}+o_{p}\left(  1\right)
\]
where $v_{t}=RQ^{-1}\tilde{X}_{t}^{\prime}u_{t}.$ Then
\[
R\hat{Q}^{-1}\hat{\Omega}Q^{-1}R^{\prime}=\frac{1}{K}\sum_{j=1}^{K}\left[
\frac{1}{\sqrt{T}}\sum_{t=1}^{T}\phi_{j}\left(  \frac{t}{T}\right)  \hat
{v}_{t}\right]  ^{\otimes2},
\]
where $\hat{v}_{t}=R\hat{Q}^{-1}\tilde{X}_{t}^{\prime}u_{t}.$ So $R\hat
{Q}^{-1}\hat{\Omega}Q^{-1}R^{\prime}$ can be viewed as the series variance
estimator of the long run variance of $\left\{  v_{t}\right\}  .$ We can
follow
\citet{Phillips2005}
and choose $K$ to minimize the mean square error (MSE) of $RQ^{-1}\hat{\Omega
}Q^{-1}R^{\prime}.$ We fit a VAR(1) model to $\hat{v}_{t}$ and use the fitted
model to compute the data-driven MSE-optimal $K.$

Table \ref{tab:data_driven_size} reports the null rejection probabilities and
the average values of $K$ used with data-driven choice of $K$ for different
sample sizes. The qualitative observations from Figures \ref{figure_ar_T100}--
\ref{Figure_ARMAT200} continue to hold with the data-driven $K.$ In
particular, the nonstandard test and standard F test are more accurate than
the corresponding chi-square tests, especially when the latter have large
positive size distortion. The null rejection probabilities of the nonstandard
test and the standard F test are close to each other. Similarly, the null
rejection probabilities of the two chi-square tests are close to each other.
As expected, the average value of $K$ decreases with the persistence of the
underlying processes. The higher the persistence, the smaller the average $K$
value, and the more effective the nonstandard test and standard F test in
reducing the size distortion.%

\begin{table}[htbp]%

\caption{The empirical null rejection probabilities of different 5\% tests with data-driven choice of $K$.}\medskip

\begin{tabular}
[c]{lcccccccc}\hline\hline
& $\rho=0$ & $\rho=0.3$ & $\rho=0.6$ & $\rho=0.9$ & $\rho=-0.6$ & $\rho=-0.3$
& $\rho=0.6$ & $\rho=0.9$\\
& $\psi=0$ & $\psi=0$ & $\psi=0$ & $\psi=0$ & $\psi=0$ & $\psi=0$ & $\psi=0.6$
& $\psi=.0.9$\\\hline
& \multicolumn{8}{c}{$T=100$}\\\cline{2-9}%
$\chi^{2}$: Fourier & 0.092 & 0.131 & 0.227 & 0.511 & 0.128 & 0.091 & 0.287 &
0.558\\
$F_{\infty}^{\ast}$: Fourier & 0.060 & 0.076 & 0.101 & 0.198 & 0.067 & 0.056 &
0.087 & 0.170\\
$\chi^{2}$: Transformed & 0.089 & 0.124 & 0.210 & 0.473 & 0.119 & 0.085 &
0.259 & 0.516\\
$F$: Transformed & 0.064 & 0.079 & 0.101 & 0.209 & 0.071 & 0.060 & 0.088 &
0.182\\
Ave$(K)$ & 30.00 & 18.40 & 9.71 & 5.29 & 16.57 & 26.34 & 6.14 & 4.27\\
&  &  &  &  &  &  &  & \\
& \multicolumn{8}{c}{$T=200$}\\\cline{2-9}%
$\chi^{2}$: Fourier & 0.069 & 0.100 & 0.153 & 0.396 & 0.092 & 0.070 & 0.197 &
0.444\\
$F_{\infty}^{\ast}$: Fourier & 0.052 & 0.066 & 0.079 & 0.150 & 0.055 & 0.050 &
0.075 & 0.131\\
$\chi^{2}$: Transformed & 0.068 & 0.094 & 0.142 & 0.363 & 0.088 & 0.067 &
0.179 & 0.406\\
$F$: Transformed & 0.057 & 0.069 & 0.082 & 0.153 & 0.058 & 0.051 & 0.074 &
0.135\\
Ave$(K)$ & 70 & 28.82 & 14.32 & 6.10 & 24.96 & 46.02 & 8.56 & 4.56\\
&  &  &  &  &  &  &  & \\
& \multicolumn{8}{c}{$T=500$}\\\cline{2-9}%
$\chi^{2}$: Fourier & 0.055 & 0.068 & 0.096 & 0.222 & 0.067 & 0.057 & 0.119 &
0.278\\
$F_{\infty}^{\ast}$: Fourier & 0.049 & 0.054 & 0.062 & 0.091 & 0.053 & 0.048 &
0.056 & 0.084\\
$\chi^{2}$: Transformed & 0.053 & 0.064 & 0.091 & 0.209 & 0.064 & 0.055 &
0.110 & 0.253\\
$F$: Transformed & 0.048 & 0.053 & 0.062 & 0.096 & 0.051 & 0.048 & 0.058 &
0.086\\
Ave$(K)$ & 144.51 & 56.47 & 26.91 & 9.03 & 46.58 & 96.41 & 15.19 &
6.05\\\hline\hline
\end{tabular}

\label{tab:data_driven_size}
\end{table}%

To simulate the power of the tests, we let $\beta_{1}=\left(  0,0\right)  $
and $\beta_{2}=\left(  \delta,\delta\right)  .$ Figure \ref{Figure_power}
presents the size-adjusted power curves as functions of $\delta$ when the
sample size is 200 and when both $q_{t}$ and $u_{t}$ follow AR(1) processes.
The figure is representative of other cases. For the two tests in each pair,
the size-adjusted powers are the same, as they are based on the same test
statistic. Thus, we need only report two power curves: one for the usual
Fourier bases and the other for the transformed Fourier bases. The basic
message from Figure \ref{Figure_power} is that the size-adjusted powers
associated with the two sets of bases are very close to each other. This,
coupled with its size accuracy and convenience to use, suggests that we use
the F test in empirical applications.%

\begin{figure}[h]%
\centering
\includegraphics[
height=6.8338in,
width=6.2578in
]%
{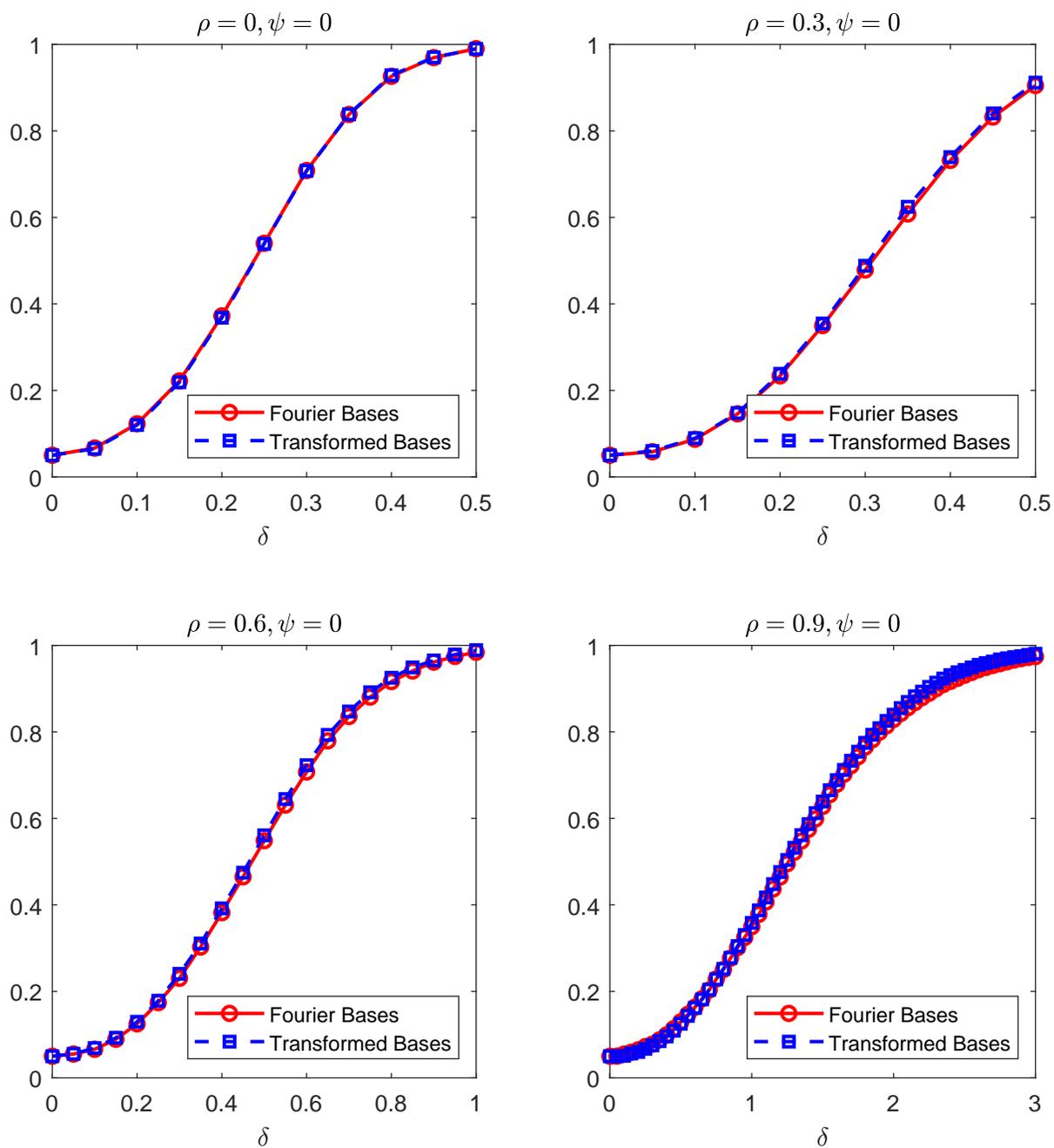}%
\caption{The size-adjusted power curves for different 5\% tests when $T=200.$}%
\label{Figure_power}%
\end{figure}

\section{Conclusion}

This study proposes asymptotic F and t tests for structural breaks that are
robust to heteroscedasticity and autocorrelation. The tests are based on a
special series HAR variance estimator where the basis functions are crafted
via the Gram--Schmidt orthonormalization. Monte Carlo simulations show that
the F test is much more accurate than the corresponding chi-square test.

This study assumes that there is a single known break point. The asymptotic F
and t theory can be extended to the case with multiple but known break points.
The theory can also be extended to allow for a linear trend or other
deterministic trends, but we need to redesign the basis functions. In
principle, the tests based on series HAR variance estimation can be extended
to accommodate the case with an unknown break point along the line of
\citet{Cho_Vogelsang2017}. All the basic ingredients have been established in
the study. We only need to take the supremum (or other functionals) of the
Wald or t statistic over $\lambda$ as the test statistic. However, the
convenient F approximation is lost, as the supremum of the standard
distributions is not standard any more. Therefore, it is not clear whether
there is still an advantage of using series HAR variance estimators rather
than kernel HAR variance estimators.

\clearpage

\section{Appendix of Proofs}

\begin{proof}
[Proof of Lemma \ref{Lemma_1}]Under Assumption \ref{Assumption XX}, we have
\[
\frac{\tilde{X}^{\prime}\tilde{X}}{T}=\left(
\begin{array}
[c]{cc}%
T^{-1}\sum_{t=1}^{[T\lambda]}X_{t}^{\prime}X_{t} & O\\
O & T^{-1}\sum_{t=[T\lambda]+1}^{T}X_{t}^{\prime}X_{t}%
\end{array}
\right)  \rightarrow^{p}\left(
\begin{array}
[c]{cc}%
\lambda Q & O\\
O & \left(  1-\lambda\right)  Q
\end{array}
\right)  .
\]
Under Assumption \ref{Assumption XU}, we have%
\[
\frac{\tilde{X}^{\prime}u}{\sqrt{T}}=\left(
\begin{array}
[c]{c}%
T^{-1}\sum_{t=1}^{[T\lambda]}X_{t}^{\prime}u_{t}\\
T^{-1}\sum_{t=[T\lambda]+1}^{T}X_{t}^{\prime}u_{t}%
\end{array}
\right)  \rightarrow^{d}\left(
\begin{array}
[c]{c}%
\Lambda W_{m}\left(  \lambda\right) \\
\Lambda\left[  W_{m}\left(  1\right)  -W_{m}\left(  \lambda\right)  \right]
\end{array}
\right)  .
\]
Hence%
\begin{align*}
\sqrt{T}(\hat{\beta}-\beta)  &  \rightarrow^{d}\left(
\begin{array}
[c]{cc}%
\lambda Q & O\\
O & \left(  1-\lambda\right)  Q
\end{array}
\right)  ^{-1}\left(
\begin{array}
[c]{c}%
\Lambda W_{m}\left(  \lambda\right) \\
\Lambda\left[  W_{m}\left(  1\right)  -W_{m}\left(  \lambda\right)  \right]
\end{array}
\right) \\
&  =\left(
\begin{array}
[c]{c}%
\left(  \lambda Q\right)  ^{-1}\Lambda W_{m}\left(  \lambda\right) \\
\left[  \left(  1-\lambda\right)  Q\right]  ^{-1}\Lambda\left[  W_{m}\left(
1\right)  -W_{m}\left(  \lambda\right)  \right]
\end{array}
\right) \\
&  =\left(
\begin{array}
[c]{c}%
Q^{-1}\Lambda\cdot\frac{W_{m}\left(  \lambda\right)  }{\lambda}\\
Q^{-1}\Lambda\cdot\frac{W_{m}\left(  1\right)  -W_{m}\left(  \lambda\right)
}{1-\lambda}%
\end{array}
\right)  =\left(
\begin{array}
[c]{c}%
Q^{-1}\Lambda\cdot\frac{1}{\lambda}\int_{0}^{\lambda}dW_{m}\left(
\lambda\right) \\
Q^{-1}\Lambda\cdot\frac{1}{1-\lambda}\int_{\lambda}^{1}dW_{m}\left(
\lambda\right)
\end{array}
\right)  .
\end{align*}

For the second part of the lemma, we have
\begin{align*}
&  \frac{1}{\sqrt{T}}\sum_{t=1}^{T}\phi_{j}\left(  \frac{t}{T}\right)
\tilde{X}_{t}^{\prime}\hat{u}_{t}\\
&  =\frac{1}{\sqrt{T}}\sum_{t=1}^{T}\phi_{j}\left(  \frac{t}{T}\right)
\tilde{X}_{t}^{\prime}\left(  Y_{t}-\tilde{X}_{t}\hat{\beta}\right) \\
&  =\frac{1}{\sqrt{T}}\sum_{t=1}^{T}\phi_{j}\left(  \frac{t}{T}\right)
\tilde{X}_{t}^{\prime}\left(  \tilde{X}_{t}\beta+u_{t}-\tilde{X}_{t}\hat
{\beta}\right) \\
&  =\frac{1}{\sqrt{T}}\sum_{t=1}^{T}\phi_{j}\left(  \frac{t}{T}\right)
\tilde{X}_{t}^{\prime}u_{t}-\left[  \frac{1}{T}\sum_{t=1}^{T}\phi_{j}\left(
\frac{t}{T}\right)  \tilde{X}_{t}^{\prime}\tilde{X}_{t}\right]  \sqrt{T}%
(\hat{\beta}-\beta).
\end{align*}
Now, it is not hard to show that under Assumption \ref{Weight_function},
\begin{align*}
&  \frac{1}{T}\sum_{t=1}^{T}\phi_{j}\left(  \frac{t}{T}\right)  \tilde{X}%
_{t}^{\prime}\tilde{X}_{t}\\
&  =\left(
\begin{array}
[c]{cc}%
\frac{1}{T}\sum_{t=1}^{[T\lambda]}\phi_{j}\left(  \frac{t}{T}\right)
X_{t}^{\prime}X_{t} & 0\\
0 & \frac{1}{T}\sum_{t=[T\lambda]+1}^{T}\phi_{j}\left(  \frac{t}{T}\right)
X_{t}^{\prime}X_{t}%
\end{array}
\right) \\
&  \rightarrow^{p}\left(
\begin{array}
[c]{cc}%
\left[  \int_{0}^{\lambda}\phi_{j}\left(  r\right)  dr\right]  Q & O\\
O & \left[  \int_{\lambda}^{1}\phi_{j}\left(  r\right)  dr\right]  Q
\end{array}
\right)  .
\end{align*}
Hence,
\begin{align*}
&  \frac{1}{\sqrt{T}}\sum_{t=1}^{T}\phi_{j}\left(  \frac{t}{T}\right)
\tilde{X}_{t}^{\prime}\hat{u}_{t}\\
&  \rightarrow^{d}\left(
\begin{array}
[c]{c}%
\Lambda\int_{0}^{\lambda}\phi_{j}\left(  r\right)  dW_{m}\left(  r\right) \\
\Lambda\int_{\lambda}^{1}\phi_{j}\left(  r\right)  dW_{m}\left(  r\right)
\end{array}
\right) \\
&  -\left(
\begin{array}
[c]{cc}%
\left[  \int_{0}^{\lambda}\phi_{j}\left(  r\right)  dr\right]  Q & O\\
O & \left[  \int_{\lambda}^{1}\phi_{j}\left(  r\right)  dr\right]  Q
\end{array}
\right)  \left(
\begin{array}
[c]{c}%
Q^{-1}\Lambda\cdot\frac{W_{m}\left(  \lambda\right)  }{\lambda}\\
Q^{-1}\Lambda\cdot\frac{W_{m}\left(  1\right)  -W_{m}\left(  \lambda\right)
}{\left(  1-\lambda\right)  }%
\end{array}
\right) \\
&  =\left(
\begin{array}
[c]{c}%
\Lambda\left\{  \int_{0}^{\lambda}\phi_{j}\left(  r\right)  dW_{m}\left(
r\right)  -\frac{1}{\lambda}\int_{0}^{\lambda}\phi_{j}\left(  r\right)
dr\cdot W_{m}\left(  \lambda\right)  \right\} \\
\Lambda\left\{  \int_{\lambda}^{1}\phi_{j}\left(  r\right)  dW_{m}\left(
r\right)  -\frac{1}{1-\lambda}\int_{\lambda}^{1}\phi_{j}\left(  r\right)
dr\left[  W_{m}\left(  1\right)  -W_{m}\left(  \lambda\right)  \right]
\right\}
\end{array}
\right) \\
&  =\left(
\begin{array}
[c]{c}%
\Lambda\int_{0}^{\lambda}\left[  \phi_{j}\left(  r\right)  -\bar{\phi}%
_{j,1}\right]  dW_{m}\left(  r\right) \\
\Lambda\int_{\lambda}^{1}\left[  \phi_{j}\left(  r\right)  -\bar{\phi}%
_{j,2}\right]  dW_{m}\left(  r\right)
\end{array}
\right)  .
\end{align*}

\end{proof}

\bigskip

\begin{proof}
[Proof of Theorem \ref{Theorem_main1}]We have%
\[
R\hat{Q}^{-1}\rightarrow^{d}\left(  \mathcal{R},-\mathcal{R}\right)  \left(
\begin{array}
[c]{cc}%
\lambda^{-1}Q^{-1} & O\\
O & \left(  1-\lambda\right)  ^{-1}Q^{-1}%
\end{array}
\right)  =\left(
\begin{array}
[c]{cc}%
\lambda^{-1}\mathcal{R}Q^{-1}, & -\left(  1-\lambda\right)  ^{-1}%
\mathcal{R}Q^{-1}%
\end{array}
\right)  ,
\]%
\[
\hat{\Omega}\rightarrow^{d}\frac{1}{K}\sum_{j=1}^{K}\left(
\begin{array}
[c]{cc}%
\Lambda & O\\
O & \Lambda
\end{array}
\right)  \left(
\begin{array}
[c]{c}%
\int_{0}^{\lambda}\left[  \phi_{j}\left(  r\right)  -\bar{\phi}_{j,1}\right]
dW_{m}\left(  r\right) \\
\int_{\lambda}^{1}\left[  \phi_{j}\left(  r\right)  -\bar{\phi}_{j,2}\right]
dW_{m}\left(  r\right)
\end{array}
\right)  ^{\otimes2}\left(
\begin{array}
[c]{cc}%
\Lambda & O\\
O & \Lambda
\end{array}
\right)  .
\]
Hence,%
\begin{align*}
&  R\hat{Q}^{-1}\hat{\Omega}Q^{-1}R^{\prime}\\
&  \rightarrow^{d}\left(
\begin{array}
[c]{cc}%
\lambda^{-1}\mathcal{R}Q^{-1}, & -\left(  1-\lambda\right)  ^{-1}%
\mathcal{R}Q^{-1}%
\end{array}
\right)  \left(
\begin{array}
[c]{cc}%
\Lambda & O\\
O & \Lambda
\end{array}
\right) \\
&  \times\frac{1}{K}\sum_{j=1}^{K}\left(
\begin{array}
[c]{c}%
\int_{0}^{\lambda}\left[  \phi_{j}\left(  r\right)  -\bar{\phi}_{j,1}\right]
dW_{m}\left(  r\right) \\
\int_{\lambda}^{1}\left[  \phi_{j}\left(  r\right)  -\bar{\phi}_{j,2}\right]
dW_{m}\left(  r\right)
\end{array}
\right)  ^{\otimes2}\\
&  \times\left(
\begin{array}
[c]{cc}%
\Lambda & O\\
O & \Lambda
\end{array}
\right)  \left(
\begin{array}
[c]{cc}%
\lambda^{-1}\left(  \mathcal{R}Q^{-1}\right)  ^{\prime} & O\\
O & \left(  1-\lambda\right)  ^{-1}\left(  \mathcal{R}Q^{-1}\right)  ^{\prime}%
\end{array}
\right) \\
&  =\left(
\begin{array}
[c]{cc}%
\lambda^{-1}\mathcal{R}Q^{-1}\Lambda, & -\left(  1-\lambda\right)
^{-1}\mathcal{R}Q^{-1}\Lambda
\end{array}
\right)  \frac{1}{K}\sum_{j=1}^{K}\left(
\begin{array}
[c]{c}%
\int_{0}^{\lambda}\left[  \phi_{j}\left(  r\right)  -\bar{\phi}_{j,1}\right]
dW_{m}\left(  r\right) \\
\int_{\lambda}^{1}\left[  \phi_{j}\left(  r\right)  -\bar{\phi}_{j,2}\right]
dW_{m}\left(  r\right)
\end{array}
\right)  ^{\otimes2}\\
&  \times\left(
\begin{array}
[c]{c}%
\lambda^{-1}\left(  \mathcal{R}Q^{-1}\Lambda\right)  ^{\prime},\\
-\left(  1-\lambda\right)  ^{-1}\left(  \mathcal{R}Q^{-1}\Lambda\right)
^{\prime}%
\end{array}
\right) \\
&  =\mathcal{R}Q^{-1}\Lambda\frac{1}{K}\sum_{j=1}^{K}\left\{  \int%
_{0}^{\lambda}\frac{\phi_{j}\left(  r\right)  -\bar{\phi}_{j,1}}{\lambda
}dW_{m}\left(  r\right)  -\int_{\lambda}^{1}\frac{\phi_{j}\left(  r\right)
-\bar{\phi}_{j,2}}{1-\lambda}dW_{m}\left(  r\right)  \right\}  ^{\otimes
2}\left(  \mathcal{R}Q^{-1}\Lambda\right)  ^{\prime}.
\end{align*}
Also, under the null, we have
\begin{align*}
\sqrt{T}\cdot\left[  R(\hat{\beta}-\beta)\right]   &  \rightarrow^{d}\left(
\mathcal{R},-\mathcal{R}\right)  \left(
\begin{array}
[c]{c}%
Q^{-1}\Lambda\cdot\frac{1}{\lambda}\int_{0}^{\lambda}dW_{m}\left(
\lambda\right) \\
Q^{-1}\Lambda\cdot\frac{1}{1-\lambda}\int_{\lambda}^{1}dW_{m}\left(
\lambda\right)
\end{array}
\right) \\
&  =\mathcal{R}Q^{-1}\Lambda\left[  \frac{1}{\lambda}\int_{0}^{\lambda}%
dW_{m}\left(  \lambda\right)  -\frac{1}{1-\lambda}\int_{\lambda}^{1}%
dW_{m}\left(  \lambda\right)  \right]  .
\end{align*}
Therefore,
\begin{align*}
&  F_{T}=T\cdot\left[  R(\hat{\beta}-\beta)\right]  ^{\prime}\left[  R\hat
{Q}^{-1}\hat{\Omega}Q^{-1}R^{\prime}\right]  ^{-1}\left[  R(\hat{\beta}%
-\beta)\right] \\
&  \rightarrow^{d}\left\{  \mathcal{R}Q^{-1}\Lambda\left[  \frac{1}{\lambda
}\int_{0}^{\lambda}dW_{m}\left(  \lambda\right)  -\frac{1}{1-\lambda}%
\int_{\lambda}^{1}dW_{m}\left(  \lambda\right)  \right]  \right\}  ^{\prime}\\
&  \times\left[  \mathcal{R}Q^{-1}\Lambda\frac{1}{K}\sum_{j=1}^{K}\left\{
\int_{0}^{1}\tilde{\phi}_{j}\left(  r;\lambda\right)  dW_{m}\left(  r\right)
\right\}  ^{\otimes2}\left(  \mathcal{R}Q^{-1}\Lambda\right)  ^{\prime
}\right]  ^{-1}\\
&  \times\mathcal{R}Q^{-1}\Lambda\left[  \frac{1}{\lambda}\int_{0}^{\lambda
}dW_{m}\left(  \lambda\right)  -\frac{1}{1-\lambda}\int_{\lambda}^{1}%
dW_{m}\left(  \lambda\right)  \right]  .
\end{align*}
Using the fact that $\mathcal{R}Q^{-1}\Lambda W_{m}=A_{p}W_{p}$ for a square
and invertible matrix $A_{p},$ we have
\begin{align*}
&  F_{T}\\
&  \rightarrow^{d}\left[  \frac{1}{\lambda}\int_{0}^{\lambda}dW_{p}\left(
\lambda\right)  -\frac{1}{1-\lambda}\int_{\lambda}^{1}dW_{p}\left(
\lambda\right)  \right]  ^{\prime}\times\left[  \frac{1}{K}\sum_{j=1}%
^{K}\left\{  \int_{0}^{1}\tilde{\phi}_{j}\left(  r;\lambda\right)
dW_{p}\left(  r\right)  \right\}  ^{\otimes2}\right]  ^{-1}\\
&  \times\left[  \frac{1}{\lambda}\int_{0}^{\lambda}dW_{p}\left(
\lambda\right)  -\frac{1}{1-\lambda}\int_{\lambda}^{1}dW_{p}\left(
\lambda\right)  \right]  .
\end{align*}

The proof for the weak convergence of $t_{T}$ is similar and is omitted to
save space.
\end{proof}

\bigskip

\begin{proof}
[Proof of Proposition \ref{Prop_F}]We prove the part for the Wald statistic
only, as the proof for the t-statistic is similar. Given that $\{\tilde{\phi
}_{j}\left(  r;\lambda\right)  \}$ are orthonormal on $L^{2}[0,1],$ we have:
\[
\eta_{j}:=\int_{0}^{1}\tilde{\phi}_{j}\left(  r;\lambda\right)  dW_{p}\left(
r\right)  \thicksim iidN(0,I_{p}).
\]
As a consequence,
\[
\sum_{j=1}^{K}\eta_{j}\eta_{j}^{\prime}\thicksim\mathbb{W}_{p}(I_{p},K),
\]
the standard Wishart distribution with degrees of freedom $K.$ So,
\[
\frac{K-p+1}{Kp}\lambda\left(  1-\lambda\right)  F_{\infty}=\frac{K-p+1}%
{Kp}\cdot\eta_{0}^{\prime}\left(  \frac{1}{K}\sum_{j=1}^{K}\eta_{j}\eta
_{j}^{\prime}\right)  ^{-1}\eta_{0}.
\]
Note that $\eta_{0},\eta_{1},\ldots,\eta_{K}$ are independent standard normal
vectors. $\eta_{0}^{\prime}\left(  \frac{1}{K}\sum_{j=1}^{K}\eta_{j}\eta
_{j}^{\prime}\right)  ^{-1}\eta_{0}$ follows Hotelling's $T^{2}$ distribution.
Using the relationship between Hotelling's $T^{2}$ distribution and the
standard $F$ distribution, we have
\[
\frac{K-p+1}{Kp}\lambda\left(  1-\lambda\right)  F_{\infty}\thicksim
F_{p,K-p+1}.
\]
It then follows that
\[
\frac{K-p+1}{Kp}\lambda\left(  1-\lambda\right)  F_{T}\rightarrow
^{d}F_{p,K-p+1}.
\]

\end{proof}

\bigskip

\begin{proof}
[Proof of Lemma \ref{Lemma: covariance}]We have%
\begin{align*}
&  \int_{0}^{1}\tilde{\phi}_{j_{1}}\left(  r;\lambda\right)  \tilde{\phi
}_{j_{2}}\left(  r;\lambda\right)  dr\\
&  =\frac{1}{\lambda^{2}}\int_{0}^{\lambda}\left[  \phi_{j_{1}}\left(
r\right)  -\frac{1}{\lambda}\int_{0}^{\lambda}\phi_{j_{1}}\left(  s\right)
ds\right]  \left[  \phi_{j_{2}}\left(  r\right)  -\frac{1}{\lambda}\int%
_{0}^{\lambda}\phi_{j_{2}}\left(  s\right)  ds\right]  dr\\
&  +\frac{1}{\left(  1-\lambda\right)  ^{2}}\int_{\lambda}^{1}\left[
\phi_{j_{1}}\left(  r\right)  -\frac{1}{1-\lambda}\int_{\lambda}^{1}%
\phi_{j_{1}}\left(  s\right)  ds\right]  \left[  \phi_{j_{2}}\left(  r\right)
-\frac{1}{1-\lambda}\int_{\lambda}^{1}\phi_{j_{2}}\left(  s\right)  ds\right]
dr,
\end{align*}
where%
\begin{align*}
&  \frac{1}{\lambda^{2}}\int_{0}^{\lambda}\left[  \phi_{j_{1}}\left(
r\right)  -\frac{1}{\lambda}\int_{0}^{\lambda}\phi_{j_{1}}\left(  s\right)
ds\right]  \left[  \phi_{j_{2}}\left(  r\right)  -\frac{1}{\lambda}\int%
_{0}^{\lambda}\phi_{j_{2}}\left(  s\right)  ds\right]  dr\\
&  =\frac{1}{\lambda^{2}}\left[  \int_{0}^{\lambda}\int_{0}^{\lambda}%
\delta(r-s)\phi_{j_{1}}\left(  r\right)  \phi_{j_{2}}\left(  s\right)
drds-\frac{1}{\lambda}\int_{0}^{\lambda}\int_{0}^{\lambda}\phi_{j_{1}}\left(
r\right)  \phi_{j_{2}}\left(  s\right)  drds\right] \\
&  =\int_{0}^{1}\int_{0}^{1}\left[  \delta(r-s)-\frac{1}{\lambda}\right]
\frac{1\left\{  \left(  r,s\right)  \in\lbrack0,\lambda]\times\lbrack
0,\lambda]\right\}  }{\lambda^{2}}\phi_{j_{1}}\left(  r\right)  \phi_{j_{2}%
}\left(  s\right)  drds,
\end{align*}
and similarly,%
\begin{align*}
&  \frac{1}{\left(  1-\lambda\right)  ^{2}}\int_{\lambda}^{1}\left[
\phi_{j_{1}}\left(  r\right)  -\frac{1}{1-\lambda}\int_{\lambda}^{1}%
\phi_{j_{1}}\left(  s\right)  ds\right]  \left[  \phi_{j_{2}}\left(  r\right)
-\frac{1}{1-\lambda}\int_{\lambda}^{1}\phi_{j_{2}}\left(  s\right)  ds\right]
dr\\
&  =\frac{1}{\left(  1-\lambda\right)  ^{2}}\int_{\lambda}^{1}\int_{\lambda
}^{1}\left[  \delta(r-s)-\frac{1}{1-\lambda}\right]  \phi_{j_{1}}\left(
r\right)  \phi_{j_{2}}\left(  s\right)  drds\\
&  =\int_{0}^{1}\int_{0}^{1}\left[  \delta(r-s)-\frac{1}{1-\lambda}\right]
\frac{1\left\{  \left(  r,s\right)  \in\lbrack\lambda,1]\times\lbrack
\lambda,1]\right\}  }{\left(  1-\lambda\right)  ^{2}}\phi_{j_{1}}\left(
r\right)  \phi_{j_{2}}\left(  s\right)  drds.
\end{align*}
Therefore,
\[
\int_{0}^{1}\tilde{\phi}_{j_{1}}\left(  r;\lambda\right)  \tilde{\phi}_{j_{2}%
}\left(  r;\lambda\right)  dr=\int_{0}^{1}\int_{0}^{1}C(r,s;\lambda
)\phi_{j_{1}}\left(  r\right)  \phi_{j_{2}}\left(  s\right)  drds.
\]

\end{proof}

\bigskip

\begin{proof}
[Proof of Theorem \ref{Theorem:General_Case}]Part (a). Under Assumption
\ref{Assumption WW}, we have%
\begin{align*}
\hat{Q}_{\tilde{X}\cdot Z}  &  =\frac{1}{T}\tilde{X}_{z}^{\prime}\tilde{X}%
_{z}=\frac{1}{T}\tilde{X}^{\prime}M_{Z}\tilde{X}\\
&  =\frac{1}{T}\tilde{X}^{\prime}\tilde{X}-\frac{1}{T}\tilde{X}^{\prime
}Z\left(  Z^{\prime}Z\right)  ^{-1}Z^{\prime}\tilde{X}\\
&  \rightarrow^{p}\left(
\begin{array}
[c]{cc}%
\lambda Q_{XX} & O\\
O & \left(  1-\lambda\right)  Q_{XX}%
\end{array}
\right)  -\left(
\begin{array}
[c]{c}%
\lambda Q_{XZ}\\
\left(  1-\lambda\right)  Q_{XZ}%
\end{array}
\right)  Q_{ZZ}^{-1}\left(
\begin{array}
[c]{cc}%
\lambda Q_{ZX} & \left(  1-\lambda\right)  Q_{ZX}%
\end{array}
\right) \\
&  :=Q_{\tilde{X}\cdot Z}.
\end{align*}
Under Assumption \ref{Assumption WU}, we have%
\begin{align*}
\frac{1}{\sqrt{T}}\tilde{X}^{\prime}M_{Z}u  &  =\frac{1}{\sqrt{T}}\tilde
{X}^{\prime}u-\frac{\tilde{X}^{\prime}Z}{T}\left(  \frac{Z^{\prime}Z}%
{T}\right)  ^{-1}\frac{Z^{\prime}u}{T}\\
&  \rightarrow^{d}\left(
\begin{array}
[c]{c}%
\Lambda_{X}W_{m+\ell}\left(  \lambda\right) \\
\Lambda_{X}\left[  W_{m+\ell}\left(  1\right)  -W_{m+\ell}\left(
\lambda\right)  \right]
\end{array}
\right)  -\left(
\begin{array}
[c]{c}%
\lambda Q_{XZ}\\
\left(  1-\lambda\right)  Q_{XZ}%
\end{array}
\right)  Q_{ZZ}^{-1}\Lambda_{Z}W_{m+\ell}\left(  1\right) \\
&  =\left(
\begin{array}
[c]{c}%
\Lambda_{X}W_{m+\ell}\left(  \lambda\right)  -\lambda\Lambda_{XZ}W_{m+\ell
}\left(  1\right) \\
\Lambda_{X}\left[  W_{m+\ell}\left(  1\right)  -W_{m+\ell}\left(
\lambda\right)  \right]  -\left(  1-\lambda\right)  \Lambda_{XZ}W_{m+\ell
}\left(  1\right)
\end{array}
\right)  .
\end{align*}
where
\[
\Lambda_{XZ}=Q_{XZ}Q_{ZZ}^{-1}\Lambda_{Z}.
\]
Hence,%
\[
R\sqrt{T}(\hat{\beta}-\beta)\rightarrow^{d}RQ_{\tilde{X}\cdot Z}^{-1}\left(
\begin{array}
[c]{c}%
\Lambda_{X}W_{m+\ell}\left(  \lambda\right)  -\lambda\Lambda_{XZ}W_{m+\ell
}\left(  1\right) \\
\Lambda_{X}\left[  W_{m+\ell}\left(  1\right)  -W_{m+\ell}\left(
\lambda\right)  \right]  -\left(  1-\lambda\right)  \Lambda_{XZ}W_{m+\ell
}\left(  1\right)
\end{array}
\right)  .
\]
Using the matrix inverse formula
\[
\left(  A-CB^{-1}C^{\prime}\right)  ^{-1}=A^{-1}+A^{-1}C\left(  B-C^{\prime
}A^{-1}C\right)  ^{-1}C^{\prime}A^{-1},
\]
we have
\begin{align*}
&  Q_{\tilde{X}\cdot Z}^{-1}\\
&  =\left(
\begin{array}
[c]{cc}%
\lambda Q_{XX} & O\\
O & \left(  1-\lambda\right)  Q_{XX}%
\end{array}
\right)  ^{-1}\\
&  +\left(
\begin{array}
[c]{cc}%
\lambda Q_{XX} & O\\
O & \left(  1-\lambda\right)  Q_{XX}%
\end{array}
\right)  ^{-1}\left(
\begin{array}
[c]{c}%
\lambda Q_{XZ}\\
\left(  1-\lambda\right)  Q_{XZ}%
\end{array}
\right) \\
&  \times\left[  Q_{ZZ}-\left(
\begin{array}
[c]{c}%
\lambda Q_{XZ}\\
\left(  1-\lambda\right)  Q_{XZ}%
\end{array}
\right)  ^{\prime}\left(
\begin{array}
[c]{cc}%
\lambda Q_{XX} & O\\
O & \left(  1-\lambda\right)  Q_{XX}%
\end{array}
\right)  ^{-1}\left(
\begin{array}
[c]{c}%
\lambda Q_{XZ}\\
\left(  1-\lambda\right)  Q_{XZ}%
\end{array}
\right)  \right]  ^{-1}\\
&  \times\left(
\begin{array}
[c]{c}%
\lambda Q_{XZ}\\
\left(  1-\lambda\right)  Q_{XZ}%
\end{array}
\right)  ^{\prime}\left(
\begin{array}
[c]{cc}%
\lambda Q_{XX} & O\\
O & \left(  1-\lambda\right)  Q_{XX}%
\end{array}
\right)  ^{-1}\\
&  =\left(
\begin{array}
[c]{cc}%
\lambda^{-1}Q_{XX}^{-1} & O\\
O & \left(  1-\lambda\right)  ^{-1}Q_{XX}^{-1}%
\end{array}
\right)  +\left(
\begin{array}
[c]{cc}%
\mathbb{Q} & \mathbb{Q}\\
\mathbb{Q} & \mathbb{Q}%
\end{array}
\right)  ,
\end{align*}
where%
\[
\mathbb{Q}=Q_{XX}^{-1}Q_{XZ}Q_{Z\cdot X}^{-1}Q_{ZX}Q_{XX}^{-1}\text{ for
}Q_{Z\cdot X}=Q_{ZZ}-Q_{ZX}Q_{XX}^{-1}Q_{XZ}.
\]
Therefore,%
\begin{equation}
RQ_{\tilde{X}\cdot Z}^{-1}=[\mathcal{R},-\mathcal{R}]\left(
\begin{array}
[c]{cc}%
\lambda^{-1}Q_{XX}^{-1} & O\\
O & \left(  1-\lambda\right)  ^{-1}Q_{XX}^{-1}%
\end{array}
\right)  =\mathcal{R}\left[  \lambda^{-1}Q_{XX}^{-1},-\left(  1-\lambda
\right)  ^{-1}Q_{XX}^{-1}\right]  \label{proof RQi}%
\end{equation}
and
\begin{align*}
&  RQ_{\tilde{X}\cdot Z}^{-1}\left(
\begin{array}
[c]{c}%
\Lambda_{X}W_{m+\ell}\left(  \lambda\right)  -\lambda\Lambda_{XZ}W_{m+\ell
}\left(  1\right) \\
\Lambda_{X}\left[  W_{m+\ell}\left(  1\right)  -W_{m+\ell}\left(
\lambda\right)  \right]  -\left(  1-\lambda\right)  \Lambda_{XZ}W_{m+\ell
}\left(  1\right)
\end{array}
\right) \\
&  =\mathcal{R}\left[  \lambda^{-1}Q_{XX}^{-1},-\left(  1-\lambda\right)
^{-1}Q_{XX}^{-1}\right]  \left(
\begin{array}
[c]{c}%
\Lambda_{X}W_{m+\ell}\left(  \lambda\right)  -\lambda\Lambda_{XZ}W_{m+\ell
}\left(  1\right) \\
\Lambda_{X}\left[  W_{m+\ell}\left(  1\right)  -W_{m+\ell}\left(
\lambda\right)  \right]  -\left(  1-\lambda\right)  \Lambda_{XZ}W_{m+\ell
}\left(  1\right)
\end{array}
\right) \\
&  =\mathcal{R}Q_{XX}^{-1}\left(  \Lambda_{X}\frac{W_{m+\ell}\left(
\lambda\right)  }{\lambda}-\Lambda_{XZ}W_{m+\ell}\left(  1\right)  \right) \\
&  -\mathcal{R}Q_{XX}^{-1}\left(  \Lambda_{X}\frac{W_{m+\ell}\left(  1\right)
-W_{m+\ell}\left(  \lambda\right)  }{1-\lambda}-\Lambda_{XZ}W_{m+\ell}\left(
1\right)  \right) \\
&  =\mathcal{R}Q_{XX}^{-1}\Lambda_{X}\left(  \frac{W_{m+\ell}\left(
\lambda\right)  }{\lambda}-\frac{W_{m+\ell}\left(  1\right)  -W_{m+\ell
}\left(  \lambda\right)  }{1-\lambda}\right)  .
\end{align*}
Hence
\begin{align*}
R\sqrt{T}(\hat{\beta}-\beta)  &  \rightarrow^{d}\mathcal{R}Q_{XX}^{-1}%
\Lambda_{X}\left(  \frac{W_{m+\ell}\left(  \lambda\right)  }{\lambda}%
-\frac{W_{m+\ell}\left(  1\right)  -W_{m+\ell}\left(  \lambda\right)
}{1-\lambda}\right) \\
&  =\mathcal{R}Q_{XX}^{-1}\Lambda_{X}\left[  \frac{1}{\lambda}\int%
_{0}^{\lambda}dW_{m+\ell}\left(  \lambda\right)  -\frac{1}{1-\lambda}%
\int_{\lambda}^{1}dW_{m+\ell}\left(  \lambda\right)  \right]  .
\end{align*}

Part (b) We have%
\begin{align*}
&  R\hat{Q}_{\tilde{X}\cdot Z}^{-1}\frac{1}{\sqrt{T}}\sum_{t=1}^{T}\phi
_{j}\left(  \frac{t}{T}\right)  \tilde{X}_{z,t}^{\prime}\hat{u}_{t}\\
&  =R\hat{Q}_{\tilde{X}\cdot Z}^{-1}\frac{1}{\sqrt{T}}\sum_{t=1}^{T}\phi
_{j}\left(  \frac{t}{T}\right)  \tilde{X}_{z,t}^{\prime}\left(  u_{t}%
-\tilde{X}_{z,t}^{\prime}(\hat{\beta}-\beta)-Z_{t}\left(  Z^{\prime}Z\right)
^{-1}Z^{\prime}u\right) \\
&  =R\hat{Q}_{\tilde{X}\cdot Z}^{-1}\frac{1}{\sqrt{T}}\sum_{t=1}^{T}\phi
_{j}\left(  \frac{t}{T}\right)  \tilde{X}_{z,t}^{\prime}u_{t}^{\ast}-R\hat
{Q}_{\tilde{X}\cdot Z}^{-1}\left[  \frac{1}{T}\sum_{t=1}^{T}\phi_{j}\left(
\frac{t}{T}\right)  \tilde{X}_{z,t}^{\prime}\tilde{X}_{z,t}\right]  \sqrt
{T}(\hat{\beta}-\beta),
\end{align*}
where $u_{t}^{\ast}=u_{t}-Z_{t}\left(  Z^{\prime}Z\right)  ^{-1}Z^{\prime}u.$
To find the limit of the first term in the above equation, we note that%
\begin{align*}
&  \frac{1}{\sqrt{T}}\sum_{t=1}^{T}\phi_{j}\left(  \frac{t}{T}\right)
\tilde{X}_{z,t}^{\prime}u_{t}^{\ast}\\
&  =\frac{1}{\sqrt{T}}\sum_{t=1}^{T}\phi_{j}\left(  \frac{t}{T}\right)
\left(  \tilde{X}_{t}-Z_{t}\left(  Z^{\prime}Z\right)  ^{-1}Z^{\prime}%
\tilde{X}\right)  ^{\prime}\left(  u_{t}-Z_{t}\left(  Z^{\prime}Z\right)
^{-1}Z^{\prime}u\right) \\
&  =\frac{1}{\sqrt{T}}\sum_{t=1}^{T}\phi_{j}\left(  \frac{t}{T}\right)
\tilde{X}_{t}^{\prime}u_{t}-\frac{1}{\sqrt{T}}\sum_{t=1}^{T}\phi_{j}\left(
\frac{t}{T}\right)  \tilde{X}_{t}^{\prime}Z_{t}\left(  Z^{\prime}Z\right)
^{-1}Z^{\prime}u\\
&  -\frac{1}{\sqrt{T}}\sum_{t=1}^{T}\phi_{j}\left(  \frac{t}{T}\right)
\tilde{X}^{\prime}Z\left(  Z^{\prime}Z\right)  ^{-1}Z_{t}^{\prime}u_{t}%
+\frac{1}{\sqrt{T}}\sum_{t=1}^{T}\phi_{j}\left(  \frac{t}{T}\right)  \tilde
{X}^{\prime}Z\left(  Z^{\prime}Z\right)  ^{-1}Z_{t}^{\prime}Z_{t}\left(
Z^{\prime}Z\right)  ^{-1}Z^{\prime}u\\
&  \rightarrow^{d}\left(
\begin{array}
[c]{c}%
\Lambda_{X}\int_{0}^{\lambda}\phi_{j}\left(  r\right)  dW_{m+\ell}\left(
r\right) \\
\Lambda_{X}\int_{\lambda}^{1}\phi_{j}\left(  r\right)  dW_{m+\ell}\left(
r\right)
\end{array}
\right)  -\left(
\begin{array}
[c]{c}%
\int_{0}^{\lambda}\phi_{j}\left(  r\right)  dr\cdot\Lambda_{XZ}W_{m+\ell
}\left(  1\right) \\
\int_{\lambda}^{1}\phi_{j}\left(  r\right)  dr\cdot\Lambda_{XZ}W_{m+\ell
}\left(  1\right)
\end{array}
\right) \\
&  -\left(
\begin{array}
[c]{c}%
\lambda\cdot\Lambda_{XZ}\int_{0}^{1}\phi_{j}\left(  r\right)  dW_{m+\ell
}\left(  r\right) \\
\left(  1-\lambda\right)  \cdot\Lambda_{XZ}\int_{0}^{1}\phi_{j}\left(
r\right)  dW_{m+\ell}\left(  r\right)
\end{array}
\right)  +\left(
\begin{array}
[c]{c}%
\lambda\bar{\phi}_{j,0}\cdot\Lambda_{XZ}W_{m+\ell}\left(  1\right) \\
\left(  1-\lambda\right)  \bar{\phi}_{j,0}\cdot\Lambda_{XZ}W_{m+\ell}\left(
1\right)
\end{array}
\right)  ,
\end{align*}
where $\bar{\phi}_{j,0}=\int_{0}^{1}\phi_{j}\left(  r\right)  dr.$ Therefore,
\begin{align}
&  R\hat{Q}_{\tilde{X}\cdot Z}^{-1}\frac{1}{\sqrt{T}}\sum_{t=1}^{T}\phi
_{j}\left(  \frac{t}{T}\right)  \tilde{X}_{z,t}^{\prime}u_{t}^{\ast
}\nonumber\\
&  \rightarrow^{d}\mathcal{R}Q_{XX}^{-1}\Lambda_{X}\left[  \lambda^{-1}%
\int_{0}^{\lambda}\phi_{j}\left(  r\right)  dW_{m+\ell}\left(  r\right)
-\left(  1-\lambda\right)  ^{-1}\int_{\lambda}^{1}\phi_{j}\left(  r\right)
dW_{m+\ell}\left(  r\right)  \right] \nonumber\\
&  -\mathcal{R}Q_{XX}^{-1}\Lambda_{XZ}\left(  \bar{\phi}_{j,1}-\bar{\phi
}_{j,2}\right)  W_{m+\ell}\left(  1\right)  , \label{RQWU}%
\end{align}
where we have used $R\hat{Q}^{-1}\rightarrow^{d}RQ_{\tilde{X},Z}%
^{-1}=\mathcal{R}\left[  \lambda^{-1}Q_{XX}^{-1},-\left(  1-\lambda\right)
^{-1}Q_{XX}^{-1}\right]  ;$ see (\ref{proof RQi}).

Next,%
\begin{align}
&  \frac{1}{T}\sum_{t=1}^{T}\phi_{j}\left(  \frac{t}{T}\right)  \tilde
{X}_{z,t}^{\prime}\tilde{X}_{z,t}\nonumber\\
&  =\frac{1}{T}\sum_{t=1}^{T}\phi_{j}\left(  \frac{t}{T}\right)  \tilde{X}%
_{t}^{\prime}\tilde{X}_{t}-\frac{1}{T}\sum_{t=1}^{T}\phi_{j}\left(  \frac
{t}{T}\right)  \tilde{X}_{t}^{\prime}Z_{t}\left(  Z^{\prime}Z\right)
^{-1}Z^{\prime}\tilde{X}\nonumber\\
&  -\frac{1}{T}\sum_{t=1}^{T}\phi_{j}\left(  \frac{t}{T}\right)  \tilde
{X}^{\prime}Z\left(  Z^{\prime}Z\right)  ^{-1}Z_{t}^{\prime}\tilde{X}%
_{t}+\frac{1}{T}\sum_{t=1}^{T}\phi_{j}\left(  \frac{t}{T}\right)  \tilde
{X}^{\prime}Z\left(  Z^{\prime}Z\right)  ^{-1}Z_{t}^{\prime}Z_{t}\left(
Z^{\prime}Z\right)  ^{-1}Z^{\prime}\tilde{X}\nonumber\\
&  \rightarrow^{p}\left(
\begin{array}
[c]{cc}%
\left[  \int_{0}^{\lambda}\phi_{j}\left(  r\right)  dr\right]  Q_{XX} & O\\
O & \left[  \int_{\lambda}^{1}\phi_{j}\left(  r\right)  dr\right]  Q_{XX}%
\end{array}
\right) \nonumber\\
&  -\left(
\begin{array}
[c]{c}%
\left[  \int_{0}^{\lambda}\phi_{j}\left(  r\right)  dr\right]  Q_{XZ}%
Q_{ZZ}^{-1}\\
\left[  \int_{\lambda}^{1}\phi_{j}\left(  r\right)  dr\right]  Q_{XZ}%
Q_{ZZ}^{-1}%
\end{array}
\right)  \left[  \lambda Q_{ZX},\left(  1-\lambda\right)  Q_{ZX}\right]
\nonumber\\
&  -\left[  \left(
\begin{array}
[c]{c}%
\left[  \int_{0}^{\lambda}\phi_{j}\left(  r\right)  dr\right]  Q_{XZ}%
Q_{ZZ}^{-1}\\
\left[  \int_{\lambda}^{1}\phi_{j}\left(  r\right)  dr\right]  Q_{XZ}%
Q_{ZZ}^{-1}%
\end{array}
\right)  \left[  \lambda Q_{ZX},\left(  1-\lambda\right)  Q_{ZX}\right]
\right]  ^{\prime}\nonumber\\
&  +\int_{0}^{1}\phi_{j}\left(  r\right)  dr\cdot\left(
\begin{array}
[c]{c}%
\lambda Q_{XZ}\\
\left(  1-\lambda\right)  Q_{XZ}%
\end{array}
\right)  Q_{ZZ}^{-1}\left(
\begin{array}
[c]{c}%
\lambda Q_{XZ}\\
\left(  1-\lambda\right)  Q_{XZ}%
\end{array}
\right)  ^{\prime}. \label{phi_w_w'}%
\end{align}
So,%
\begin{align*}
&  R\hat{Q}_{\tilde{X}\cdot Z}^{-1}\left[  \frac{1}{T}\sum_{t=1}^{T}\phi
_{j}\left(  \frac{t}{T}\right)  \tilde{X}_{z,t}^{\prime}\tilde{X}_{z,t}\right]
\\
&  \rightarrow^{p}\left[  \lambda^{-1}\mathcal{R}Q_{XX}^{-1},-\left(
1-\lambda\right)  ^{-1}\mathcal{R}Q_{XX}^{-1}\right]  \left(
\begin{array}
[c]{cc}%
\left[  \int_{0}^{\lambda}\phi_{j}\left(  r\right)  dr\right]  Q_{XX} & O\\
O & \left[  \int_{\lambda}^{1}\phi_{j}\left(  r\right)  dr\right]  Q_{XX}%
\end{array}
\right) \\
&  -\left[  \lambda^{-1}\mathcal{R}Q_{XX}^{-1},-\left(  1-\lambda\right)
^{-1}\mathcal{R}Q_{XX}^{-1}\right]  \left(
\begin{array}
[c]{c}%
\left[  \int_{0}^{\lambda}\phi_{j}\left(  r\right)  dr\right]  Q_{XZ}%
Q_{ZZ}^{-1}\\
\left[  \int_{\lambda}^{1}\phi_{j}\left(  r\right)  dr\right]  Q_{XZ}%
Q_{ZZ}^{-1}%
\end{array}
\right)  \left(  \lambda Q_{ZX},\left(  1-\lambda\right)  Q_{ZX}\right) \\
&  =\left[  \bar{\phi}_{j,1}\mathcal{R},-\bar{\phi}_{j,2}\mathcal{R}\right]
-\left(  \bar{\phi}_{j,1}-\bar{\phi}_{j,2}\right)  \mathcal{R}Q_{XX}%
^{-1}Q_{XZ}Q_{ZZ}^{-1}\cdot\left[  \lambda Q_{ZX},\left(  1-\lambda\right)
Q_{ZX}\right]  ,
\end{align*}
where we have used the fact that the last two terms in (\ref{phi_w_w'})
pre-multiplied by $\left[  \lambda^{-1}\mathcal{R}Q_{XX}^{-1},-\left(
1-\lambda\right)  ^{-1}\mathcal{R}Q_{XX}^{-1}\right]  $ are equal to zero.

It then follows that
\begin{align*}
&  R\hat{Q}_{\tilde{X}\cdot Z}^{-1}\left[  \frac{1}{T}\sum_{t=1}^{T}\phi
_{j}\left(  \frac{t}{T}\right)  \tilde{X}_{z,t}^{\prime}\tilde{X}%
_{z,t}\right]  \hat{Q}_{\tilde{X}\cdot Z}^{-1}\\
&  \rightarrow^{p}\left[  \bar{\phi}_{j,1}\mathcal{R},-\bar{\phi}%
_{j,2}\mathcal{R}\right]  \left\{  \left(
\begin{array}
[c]{cc}%
\lambda^{-1}Q_{XX}^{-1} & O\\
O & \left(  1-\lambda\right)  ^{-1}Q_{XX}^{-1}%
\end{array}
\right)  +\left(
\begin{array}
[c]{cc}%
\mathbb{Q} & \mathbb{Q}\\
\mathbb{Q} & \mathbb{Q}%
\end{array}
\right)  \right\} \\
&  -\left(  \bar{\phi}_{j,1}-\bar{\phi}_{j,2}\right)  \mathcal{R}Q_{XX}%
^{-1}Q_{XZ}Q_{ZZ}^{-1}\\
&  \times\left[
\begin{array}
[c]{cc}%
\lambda Q_{ZX}, & \left(  1-\lambda\right)  Q_{ZX}%
\end{array}
\right]  \left\{  \left(
\begin{array}
[c]{cc}%
\lambda^{-1}Q_{XX}^{-1} & O\\
O & \left(  1-\lambda\right)  ^{-1}Q_{XX}^{-1}%
\end{array}
\right)  +\left(
\begin{array}
[c]{cc}%
\mathbb{Q} & \mathbb{Q}\\
\mathbb{Q} & \mathbb{Q}%
\end{array}
\right)  \right\} \\
&  =\left[
\begin{array}
[c]{cc}%
\lambda^{-1}\bar{\phi}_{j,1}\mathcal{R}Q_{XX}^{-1}, & -\left(  1-\lambda
\right)  ^{-1}\bar{\phi}_{j,2}\mathcal{R}Q_{XX}^{-1}%
\end{array}
\right] \\
&  -\left(  \bar{\phi}_{j,1}-\bar{\phi}_{j,2}\right)  \mathcal{R}Q_{XX}%
^{-1}Q_{XZ}Q_{ZZ}^{-1}\left[
\begin{array}
[c]{cc}%
Q_{ZX}Q_{XX}^{-1}, & Q_{ZX}Q_{XX}^{-1}%
\end{array}
\right] \\
&  -\left(  \bar{\phi}_{j,1}-\bar{\phi}_{j,2}\right)  \mathcal{R}\left[
Q_{XX}^{-1}Q_{XZ}Q_{ZZ}^{-1}Q_{ZX}-I_{p}\right]  \left[  \mathbb{Q}%
,\mathbb{Q}\right]  .
\end{align*}
As a consequence,
\begin{equation}
R\hat{Q}_{\tilde{X}\cdot Z}^{-1}\left[  \frac{1}{T}\sum_{t=1}^{T}\phi
_{j}\left(  \frac{t}{T}\right)  \tilde{X}_{z,t}^{\prime}\tilde{X}%
_{z,t}\right]  \sqrt{T}(\hat{\beta}-\beta)\rightarrow^{d}I_{1}+I_{2}+I_{3},
\label{I1_I2_I3}%
\end{equation}
where
\begin{align*}
I_{1}  &  =\left[
\begin{array}
[c]{cc}%
\lambda^{-1}\bar{\phi}_{j,1}\mathcal{R}Q_{XX}^{-1}, & -\left(  1-\lambda
\right)  ^{-1}\bar{\phi}_{j,2}\mathcal{R}Q_{XX}^{-1}%
\end{array}
\right] \\
&  \times\left(
\begin{array}
[c]{c}%
\Lambda_{X}W_{m+\ell}\left(  \lambda\right)  -\lambda\Lambda_{XZ}W_{m+\ell
}\left(  1\right) \\
\Lambda_{X}\left[  W_{m+\ell}\left(  1\right)  -W_{m+\ell}\left(
\lambda\right)  \right]  -\left(  1-\lambda\right)  \Lambda_{XZ}W_{m+\ell
}\left(  1\right)
\end{array}
\right) \\
&  =\lambda^{-1}\bar{\phi}_{j,1}\mathcal{R}Q_{XX}^{-1}\Lambda_{X}W_{m+\ell
}\left(  \lambda\right)  -\bar{\phi}_{j,1}\mathcal{R}Q_{XX}^{-1}\Lambda
_{XZ}W_{m+\ell}\left(  1\right) \\
&  -\left(  1-\lambda\right)  ^{-1}\bar{\phi}_{j,2}\mathcal{R}Q_{XX}%
^{-1}\Lambda_{X}\left[  W_{m+\ell}\left(  1\right)  -W_{m+\ell}\left(
\lambda\right)  \right]  +\bar{\phi}_{j,2}\mathcal{R}Q_{XX}^{-1}\Lambda
_{XZ}W_{m+\ell}\left(  1\right) \\
&  =\left(  \lambda^{-1}\bar{\phi}_{j,1}+\left(  1-\lambda\right)  ^{-1}%
\bar{\phi}_{j,2}\right)  \mathcal{R}Q_{XX}^{-1}\Lambda_{X}W_{m+\ell}\left(
\lambda\right) \\
&  -\left[  \left(  \bar{\phi}_{j,1}-\bar{\phi}_{j,2}\right)  \mathcal{R}%
Q_{XX}^{-1}\Lambda_{XZ}+\left(  1-\lambda\right)  ^{-1}\bar{\phi}%
_{j,2}\mathcal{R}Q_{XX}^{-1}\Lambda_{X}\right]  W_{m+\ell}\left(  1\right) \\
&  =\left(  \lambda^{-1}\bar{\phi}_{j,1}+\left(  1-\lambda\right)  ^{-1}%
\bar{\phi}_{j,2}\right)  \mathcal{R}Q_{XX}^{-1}\Lambda_{X}W_{m+\ell}\left(
\lambda\right) \\
&  +\left[  \left(  \bar{\phi}_{j,1}-\bar{\phi}_{j,2}\right)  \mathcal{R}%
Q_{XX}^{-1}\left(  \Lambda_{X}-\Lambda_{XZ}\right)  \right]  W_{m+\ell}\left(
1\right) \\
&  -\left\{  \left(  1-\lambda\right)  ^{-1}\bar{\phi}_{j,2}+\left[  \left(
\bar{\phi}_{j,1}-\bar{\phi}_{j,2}\right)  \right]  \right\}  \mathcal{R}%
Q_{XX}^{-1}\Lambda_{X}W_{m+\ell}\left(  1\right) \\
&  =\left(  \lambda^{-1}\bar{\phi}_{j,1}+\left(  1-\lambda\right)  ^{-1}%
\bar{\phi}_{j,2}\right)  \mathcal{R}Q_{XX}^{-1}\Lambda_{X}W_{m+\ell}\left(
\lambda\right) \\
&  +\left[  \left(  \bar{\phi}_{j,1}-\bar{\phi}_{j,2}\right)  \mathcal{R}%
Q_{XX}^{-1}\left(  \Lambda_{X}-\Lambda_{XZ}\right)  \right]  W_{m+\ell}\left(
1\right) \\
&  -\lambda\left\{  \lambda^{-1}\bar{\phi}_{j,1}+\left(  1-\lambda\right)
^{-1}\bar{\phi}_{j,2}\right\}  \mathcal{R}Q_{XX}^{-1}\Lambda_{X}W_{m+\ell
}\left(  1\right)  ,
\end{align*}%
\begin{align*}
&  I_{2}=-\left(  \bar{\phi}_{j,1}-\bar{\phi}_{j,2}\right)  \mathcal{R}%
Q_{XX}^{-1}Q_{XZ}Q_{ZZ}^{-1}\\
&  \times\left[
\begin{array}
[c]{cc}%
Q_{ZX}Q_{XX}^{-1}, & Q_{ZX}Q_{XX}^{-1}%
\end{array}
\right]  \left(
\begin{array}
[c]{c}%
\Lambda_{X}W_{m+\ell}\left(  \lambda\right)  -\lambda\Lambda_{XZ}W_{m+\ell
}\left(  1\right) \\
\Lambda_{X}\left[  W_{m+\ell}\left(  1\right)  -W_{m+\ell}\left(
\lambda\right)  \right]  -\left(  1-\lambda\right)  \Lambda_{XZ}W_{m+\ell
}\left(  1\right)
\end{array}
\right) \\
&  =-\left(  \bar{\phi}_{j,1}-\bar{\phi}_{j,2}\right)  \mathcal{R}Q_{XX}%
^{-1}Q_{XZ}Q_{ZZ}^{-1}\cdot Q_{ZX}Q_{XX}^{-1}\left[  \Lambda_{X}-Q_{XZ}%
Q_{ZZ}^{-1}\Lambda_{Z}\right]  W_{m+\ell}\left(  1\right)  ,
\end{align*}
and%
\begin{align*}
I_{3}  &  =-\left(  \bar{\phi}_{j,1}-\bar{\phi}_{j,2}\right)  \mathcal{R}%
\left[  Q_{XX}^{-1}Q_{XZ}Q_{ZZ}^{-1}Q_{ZX}-I_{p}\right]  \left[
\mathbb{Q},\mathbb{Q}\right] \\
&  \times\left(
\begin{array}
[c]{c}%
\Lambda_{X}W_{m+\ell}\left(  \lambda\right)  -\lambda\Lambda_{XZ}W_{m+\ell
}\left(  1\right) \\
\Lambda_{X}\left[  W_{m+\ell}\left(  1\right)  -W_{m+\ell}\left(
\lambda\right)  \right]  -\left(  1-\lambda\right)  \Lambda_{XZ}W_{m+\ell
}\left(  1\right)
\end{array}
\right) \\
&  =-\left(  \bar{\phi}_{j,1}-\bar{\phi}_{j,2}\right)  \mathcal{R}\left[
Q_{XX}^{-1}Q_{XZ}Q_{ZZ}^{-1}Q_{ZX}-I_{p}\right]  \mathbb{Q}\left[  \Lambda
_{X}-Q_{XZ}Q_{ZZ}^{-1}\Lambda_{Z}\right]  W_{m+\ell}\left(  1\right)  .
\end{align*}
Plugging the above three terms $I_{1},I_{2},$ and $I_{3}$ back into
(\ref{I1_I2_I3}), we obtain:%
\begin{align}
&  R\hat{Q}_{\tilde{X}\cdot Z}^{-1}\left[  \frac{1}{T}\sum_{t=1}^{T}\phi
_{j}\left(  \frac{t}{T}\right)  \tilde{X}_{z,t}^{\prime}\tilde{X}%
_{z,t}\right]  \sqrt{T}\left(  \hat{\beta}-\beta\right) \nonumber\\
&  \rightarrow^{d}\left(  \lambda^{-1}\bar{\phi}_{j,1}+\left(  1-\lambda
\right)  ^{-1}\bar{\phi}_{j,2}\right)  \mathcal{R}Q_{XX}^{-1}\Lambda
_{X}W_{m+\ell}\left(  \lambda\right) \nonumber\\
&  -\lambda\left(  \lambda^{-1}\bar{\phi}_{j,1}+\left(  1-\lambda\right)
^{-1}\bar{\phi}_{j,2}\right)  \mathcal{R}Q_{XX}^{-1}\Lambda_{X}W_{m+\ell
}\left(  1\right) \nonumber\\
&  +\left(  \bar{\phi}_{j,1}-\bar{\phi}_{j,2}\right)  \mathcal{R}Q_{XX}%
^{-1}\left[  I_{p}-Q_{XZ}Q_{ZZ}^{-1}Q_{ZX}Q_{XX}^{-1}-\left(  Q_{XZ}%
Q_{ZZ}^{-1}Q_{ZX}-Q_{XX}\right)  \mathbb{Q}\right] \nonumber\\
&  \times\left(  \Lambda_{X}-Q_{XZ}Q_{ZZ}^{-1}\Lambda_{Z}\right)  W_{m+\ell
}\left(  1\right) \nonumber\\
&  =\left(  \lambda^{-1}\bar{\phi}_{j,1}+\left(  1-\lambda\right)  ^{-1}%
\bar{\phi}_{j,2}\right)  \mathcal{R}Q_{XX}^{-1}\Lambda_{X}\left[  W_{m+\ell
}\left(  \lambda\right)  -\lambda W_{m+\ell}\left(  1\right)  \right]
\nonumber\\
&  +\left(  \bar{\phi}_{j,1}-\bar{\phi}_{j,2}\right)  \mathcal{R}Q_{XX}%
^{-1}\left(  \Lambda_{X}-Q_{XZ}Q_{ZZ}^{-1}\Lambda_{Z}\right)  W_{m+\ell
}\left(  1\right)  , \label{RQWW_beta}%
\end{align}
where the last equality holds because%
\begin{align*}
&  Q_{XX}^{-1}\left[  I_{p}-Q_{XZ}Q_{ZZ}^{-1}Q_{ZX}Q_{XX}^{-1}-\left(
Q_{XZ}Q_{ZZ}^{-1}Q_{ZX}-Q_{XX}\right)  \mathbb{Q}\right] \\
&  =Q_{XX}^{-1}-Q_{XX}^{-1}Q_{XZ}Q_{ZZ}^{-1}Q_{ZX}Q_{XX}^{-1}-Q_{XX}%
^{-1}Q_{XZ}Q_{ZZ}^{-1}Q_{ZX}Q_{XX}^{-1}Q_{XZ}Q_{Z\cdot X}^{-1}Q_{ZX}%
Q_{XX}^{-1}\\
&  =Q_{XX}^{-1}-Q_{XX}^{-1}Q_{XZ}Q_{ZZ}^{-1}Q_{ZX}Q_{XX}^{-1}\\
&  +Q_{XX}^{-1}Q_{XZ}Q_{ZZ}^{-1}\left(  -Q_{ZX}Q_{XX}^{-1}Q_{XZ}\right)
\left(  Q_{ZZ}-Q_{ZX}Q_{XX}^{-1}Q_{XZ}\right)  ^{-1}Q_{ZX}Q_{XX}%
^{-1}+\mathbb{Q}\\
&  =Q_{XX}^{-1}-Q_{XX}^{-1}Q_{XZ}Q_{ZZ}^{-1}Q_{ZX}Q_{XX}^{-1}\\
&  +Q_{XX}^{-1}Q_{XZ}Q_{ZZ}^{-1}\left(  Q_{ZZ}-Q_{ZX}Q_{XX}^{-1}Q_{XZ}\right)
\left(  Q_{ZZ}-Q_{ZX}Q_{XX}^{-1}Q_{XZ}\right)  ^{-1}Q_{ZX}Q_{XX}^{-1}\\
&  -Q_{XX}^{-1}Q_{XZ}Q_{ZZ}^{-1}Q_{ZZ}\left(  Q_{ZZ}-Q_{ZX}Q_{XX}^{-1}%
Q_{XZ}\right)  ^{-1}Q_{ZX}Q_{XX}^{-1}+\mathbb{Q}\\
&  =Q_{XX}^{-1}-Q_{XX}^{-1}Q_{XZ}Q_{ZZ}^{-1}Q_{ZX}Q_{XX}^{-1}+Q_{XX}%
^{-1}Q_{XZ}Q_{ZZ}^{-1}Q_{ZX}Q_{XX}^{-1}\\
&  -Q_{XX}^{-1}Q_{XZ}\left(  Q_{ZZ}-Q_{ZX}Q_{XX}^{-1}Q_{XZ}\right)
^{-1}Q_{ZX}Q_{XX}^{-1}+\mathbb{Q}\\
&  =Q_{XX}^{-1}-Q_{XX}^{-1}Q_{XZ}\left(  Q_{ZZ}-Q_{ZX}Q_{XX}^{-1}%
Q_{XZ}\right)  ^{-1}Q_{ZX}Q_{XX}^{-1}+Q_{XX}^{-1}Q_{XZ}Q_{Z\cdot X}^{-1}%
Q_{ZX}Q_{XX}^{-1}\\
&  =Q_{XX}^{-1}.
\end{align*}
Using (\ref{RQWU}) and (\ref{RQWW_beta}), we have
\begin{align*}
&  R\hat{Q}_{\tilde{X}\cdot Z}^{-1}\frac{1}{\sqrt{T}}\sum_{t=1}^{T}\phi
_{j}\left(  \frac{t}{T}\right)  \tilde{X}_{z,t}^{\prime}\hat{u}_{t}\\
&  \rightarrow^{d}\mathcal{R}Q_{XX}^{-1}\Lambda_{X}\left[  \lambda^{-1}%
\int_{0}^{\lambda}\phi_{j}\left(  r\right)  dW_{m+\ell}\left(  r\right)
-\left(  1-\lambda\right)  ^{-1}\int_{\lambda}^{1}\phi_{j}\left(  r\right)
dW_{m+\ell}\left(  r\right)  \right] \\
&  -\mathcal{R}Q_{XX}^{-1}\Lambda_{X}\left[  \lambda^{-1}\bar{\phi}%
_{j,1}+\left(  1-\lambda\right)  ^{-1}\bar{\phi}_{j,2}\right]  W_{m+\ell
}\left(  \lambda\right) \\
&  +\mathcal{R}Q_{XX}^{-1}\Lambda_{X}\left[  \bar{\phi}_{j,1}+\lambda\left(
1-\lambda\right)  ^{-1}\bar{\phi}_{j,2}-\left(  \bar{\phi}_{j,1}-\bar{\phi
}_{j,2}\right)  \right]  W_{m+\ell}\left(  1\right) \\
&  =\mathcal{R}Q_{XX}^{-1}\Lambda_{X}\left[  \lambda^{-1}\int_{0}^{\lambda
}\phi_{j}\left(  r\right)  dW_{m+\ell}\left(  r\right)  -\left(
1-\lambda\right)  ^{-1}\int_{\lambda}^{1}\phi_{j}\left(  r\right)  dW_{m+\ell
}\left(  r\right)  \right] \\
&  -\mathcal{R}Q_{XX}^{-1}\Lambda_{X}\left[  \lambda^{-1}\bar{\phi}%
_{j,1}+\left(  1-\lambda\right)  ^{-1}\bar{\phi}_{j,2}\right]  W_{m+\ell
}\left(  \lambda\right) \\
&  +\mathcal{R}Q_{XX}^{-1}\Lambda_{X}\left[  \left(  1-\lambda\right)
^{-1}\bar{\phi}_{j,2}\right]  W_{m+\ell}\left(  1\right) \\
&  =\mathcal{R}Q_{XX}^{-1}\Lambda_{X}\left[  \lambda^{-1}\int_{0}^{\lambda
}\left(  \phi_{j}\left(  r\right)  -\bar{\phi}_{j,1}\right)  dW_{m+\ell
}\left(  r\right)  -\left(  1-\lambda\right)  ^{-1}\int_{\lambda}^{1}\left(
\phi_{j}\left(  r\right)  -\bar{\phi}_{j,2}\right)  dW_{m+\ell}\left(
r\right)  \right] \\
&  =\mathcal{R}Q_{XX}^{-1}\Lambda_{X}\int_{0}^{1}\tilde{\phi}_{j}\left(
r;\lambda\right)  dW_{m+\ell}\left(  r\right)  .
\end{align*}

Part (c). We prove the case for $F_{T}$ only, as the proof for $t_{T}$ is
similar. Using Parts (a) and (b), we have
\begin{align*}
F_{T}  &  =T\cdot(R\hat{\beta})^{\prime}\left[  R\hat{Q}^{-1}\hat{\Omega
}Q^{-1}R^{\prime}\right]  ^{-1}R\hat{\beta}\\
&  \rightarrow^{d}\left[  \mathcal{R}Q_{XX}^{-1}\Lambda_{X}\left(  \frac
{1}{\lambda}\int_{0}^{\lambda}dW_{m+\ell}\left(  \lambda\right)  -\frac
{1}{1-\lambda}\int_{\lambda}^{1}dW_{m+\ell}\left(  \lambda\right)  \right)
\right]  ^{\prime}\\
&  \times\left\{  \sum_{j=1}^{K}\left[  \mathcal{R}Q_{XX}^{-1}\Lambda\int%
_{0}^{1}\tilde{\phi}_{j}\left(  r;\lambda\right)  dW_{m+\ell}\left(  r\right)
\right]  ^{\otimes2}\right\}  ^{-1}\\
&  \times\left[  \mathcal{R}Q_{XX}^{-1}\Lambda_{X}\left(  \frac{1}{\lambda
}\int_{0}^{\lambda}dW_{m+\ell}\left(  \lambda\right)  -\frac{1}{1-\lambda}%
\int_{\lambda}^{1}dW_{m+\ell}\left(  \lambda\right)  \right)  \right]  .
\end{align*}
Note that
\[
\mathcal{R}Q_{XX}^{-1}\Lambda_{X}W_{m+\ell}\left(  \lambda\right)
=^{d}\mathcal{A}_{p}W_{p}\left(  \lambda\right)
\]
for a $p\times p$ invertible matrix $\mathcal{A}_{p}$ such that $\mathcal{A}%
_{p}^{\prime}\mathcal{A}_{p}=\mathcal{R}Q_{XX}^{-1}\Lambda_{X}\Lambda
_{X}^{\prime}Q_{XX}^{-1}\mathcal{R}^{\prime}.$ Using this distributional
equivalence, we have%
\begin{align*}
F_{T}  &  =T\cdot(R\hat{\beta})^{\prime}\left[  R\hat{Q}^{-1}\hat{\Omega
}Q^{-1}R^{\prime}\right]  ^{-1}R\hat{\beta}\\
&  \rightarrow^{d}\left[  \mathcal{A}_{p}\left(  \frac{1}{\lambda}\int%
_{0}^{\lambda}dW_{p}\left(  \lambda\right)  -\frac{1}{1-\lambda}\int_{\lambda
}^{1}dW_{p}\left(  \lambda\right)  \right)  \right]  ^{\prime}\times\left\{
\frac{1}{K}\sum_{j=1}^{K}\left[  \mathcal{A}_{p}\int_{0}^{1}\tilde{\phi}%
_{j}\left(  r;\lambda\right)  dW_{p}\left(  r\right)  \right]  ^{\otimes
2}\right\}  ^{-1}\\
&  \times\left[  \mathcal{A}_{p}\left(  \frac{1}{\lambda}\int_{0}^{\lambda
}dW_{p}\left(  \lambda\right)  -\frac{1}{1-\lambda}\int_{\lambda}^{1}%
dW_{p}\left(  \lambda\right)  \right)  \right] \\
&  =\left(  \frac{1}{\lambda}\int_{0}^{\lambda}dW_{p}\left(  \lambda\right)
-\frac{1}{1-\lambda}\int_{\lambda}^{1}dW_{p}\left(  \lambda\right)  \right)
^{\prime}\times\left\{  \frac{1}{K}\sum_{j=1}^{K}\left[  \int_{0}^{1}%
\tilde{\phi}_{j}\left(  r;\lambda\right)  dW_{p}\left(  r\right)  \right]
^{\otimes2}\right\}  ^{-1}\\
&  \times\left(  \frac{1}{\lambda}\int_{0}^{\lambda}dW_{p}\left(
\lambda\right)  -\frac{1}{1-\lambda}\int_{\lambda}^{1}dW_{p}\left(
\lambda\right)  \right)  ,
\end{align*}
as desired.
\end{proof} 

\bibliographystyle{apalike}
\bibliography{references}

\begin{thebibliography}{}

\bibitem[Andrews, 1991]{Andrews1991}
Andrews, D. W.~K. (1991).
\newblock Heteroskedasticity and autocorrelation consistent covariance matrix
  estimation.
\newblock {\em Econometrica}, 59:817--858.

\bibitem[Cho and Vogelsang, 2017]{Cho_Vogelsang2017}
Cho, C.-K. and Vogelsang, T.~J. (2017).
\newblock Fixed-b inference for testing structural change in a time series
  regression.
\newblock {\em Econometrics}, 5(1).

\bibitem[Chow, 1960]{CHOW1960}
Chow, G.~C. (1960).
\newblock Tests of equality between sets of coefficients in two linear
  regressions.
\newblock {\em Econometrica}, 28(3):591--605.

\bibitem[Giles and Scott, 1992]{GILES1992}
Giles, D. and Scott, M. (1992).
\newblock Some consequences of using the {Chow} test in the context of
  autocorrelated disturbances.
\newblock {\em Economics Letters}, 38(2):145 -- 150.

\bibitem[Hwang and Sun, 2017]{HS2017}
Hwang, J. and Sun, Y. (2017).
\newblock Asymptotic {F} and t tests in an efficient {GMM} setting.
\newblock {\em Journal of Econometrics}, 198:277--295.

\bibitem[Jansson, 2004]{J2004}
Jansson, M. (2004).
\newblock On the error of rejection probability in simple autocorrelation
  robust tests.
\newblock {\em Econometrica}, 72:937--946.

\bibitem[Kiefer and Vogelsang, 2002a]{KV2002a}
Kiefer, N.~M. and Vogelsang, T.~J. (2002a).
\newblock {\GG{2002a} }{Heteroskedasticity-autocorrelation} robust testing
  using bandwidth equal to sample size.
\newblock {\em Econometric Theory}, 18:1350--1366.

\bibitem[Kiefer and Vogelsang, 2002b]{KV2002b}
Kiefer, N.~M. and Vogelsang, T.~J. (2002b).
\newblock {\GG{2002b}}{ Heteroskedasticity-autocorrelation} robust standard
  errors using the { Bartlett} kernel without truncation.
\newblock {\em Econometrica}, 70:2093--2095.

\bibitem[Kiefer and Vogelsang, 2005]{KV2005}
Kiefer, N.~M. and Vogelsang, T.~J. (2005).
\newblock A new asymptotic theory for heteroskedasticity-autocorrelation robust
  tests.
\newblock {\em Econometric Theory}, 21:1130--1164.

\bibitem[Kr{\"a}mer, 1989]{KRAMER1989}
Kr{\"a}mer, W. (1989).
\newblock {\em The Robustness of the Chow Test to Autocorrelation among
  Disturbances}, pages 45--52.
\newblock in Statistical Analysis and Forecasting of Economic Structural
  Change, Springer.

\bibitem[Lazarus et~al., 2018]{LLSW2018}
Lazarus, E., Lewis, D.~J., Stock, J.~H., and Watson, M.~W. (2018).
\newblock {HAR} inference: Recommendations for practice.
\newblock {\em Journal of Business \& Economic Statistics}, 36(4):541--559.

\bibitem[Liu and Sun, 2019]{LS2018}
Liu, C. and Sun, Y. (2019).
\newblock A simple and trustworthy asymptotic t test in
  difference-in-differences regressions.
\newblock {\em Journal of Econometrics}, 210(2):327 -- 362.

\bibitem[Martínez-Iriarte et~al., 2019]{MSW2019}
Martínez-Iriarte, J., Sun, Y., and Wang, X. (2019).
\newblock Asymptotic {F} tests under possibly weak identification.
\newblock Working Paper, Department of Economics, UC San Diego.

\bibitem[Newey and West, 1987]{NW1987}
Newey, W.~K. and West, K.~D. (1987).
\newblock A simple, positive semi-definite, heteroskedasticity and
  autocorrelation consistent covariance matrix.
\newblock {\em Econometrica}, 55(3):703--708.

\bibitem[Phillips, 2005]{Phillips2005}
Phillips, P. C.~B. (2005).
\newblock {HAC} estimation by automated regression.
\newblock {\em Econometric Theory}, 21(1):116--142.

\bibitem[Sun, 2011]{S2011}
Sun, Y. (2011).
\newblock Robust trend inference with series variance estimator and
  testing-optimal smoothing parameter.
\newblock {\em Journal of Econometrics}, 164:345--366.

\bibitem[Sun, 2013]{S2013}
Sun, Y. (2013).
\newblock A heteroskedasticity and autocorrelation robust {F} test using
  orthonormal series variance estimator.
\newblock {\em Econometrics Journal}, 16:1--26.

\bibitem[Sun and Kim, 2012]{SK2012}
Sun, Y. and Kim, M.~S. (2012).
\newblock Simple and powerful {GMM} over-identification tests with accurate
  size.
\newblock {\em Journal of Econometrics}, 166:267--281.

\bibitem[Sun et~al., 2008]{SPJ2008}
Sun, Y., Philips, P. C.~B., and Jin, S. (2008).
\newblock Optimal bandwidth selection in heteroskedasticity-autocorrelation
  robust testing.
\newblock {\em Econometrica}, 76:175--194.

\bibitem[Wang and Sun, 2019]{WS2019}
Wang, X. and Sun, Y. (2019).
\newblock An asymptotic {F} test for uncorrelatedness in the presence of time
  series dependence.
\newblock Working Paper, Department of Economics, UC San Diego.

\end{thebibliography}

\end{document}